\documentclass[a4paper,english,thm-restate]{lipics-v2021}
\usepackage{graphicx}
\usepackage{amsfonts}
\usepackage{enumitem}
\usepackage{caption}
\usepackage[title]{appendix}
\usepackage{pifont}
\usepackage{microtype} 
\usepackage{xspace}
\usepackage[dvipsnames]{xcolor}
\usepackage{tcolorbox}
\usepackage{hyperref}

\graphicspath{{./figures/}}
\hideLIPIcs

\newcommand{\mkmcal}[1]{\ensuremath{\mathcal{#1}}\xspace}
\newcommand{\T}{\mkmcal{T}}

\renewcommand{\P}{\mkmcal{P}}
\newcommand{\F}{\mkmcal{F}}
\renewcommand{\H}{\mkmcal{H}}
\newcommand{\X}{\mkmcal{X}}
\newcommand {\R} {\mathbb{R}}

\newcommand{\env}{\mathrm{env}\xspace}%
\newcommand{\envR}{\env}
\newcommand{\envM}{\envR^{\leq\delta}}
\newcommand{\nextE}{\mathrm{next}\xspace}%

\newcommand{\geod}{\pi\xspace}

\newcommand{\self}[1]{{#1^-}\xspace}
\newcommand{\twin}[1]{{#1^+}\xspace}
\newcommand{\eps}{\ensuremath{\varepsilon}\xspace}
\newcommand{\etal}{et al.}

\newcommand{\enumit}[1]{\textcolor{darkgray}{\sffamily\bfseries\upshape\mathversion{bold}#1}}
\newcommand{\algcall}[1]{\enumit{\textsc{#1}}}

\DeclareMathOperator{\SPM}{SPM}

    \newcommand{\from}{\colon}
    \newcommand{\Img}{\mathrm{Im}}
    
    \usepackage{tikz}
    \usetikzlibrary{calc,decorations.pathmorphing,shapes}
    
    \newcounter{sarrow}
    \newcommand\xrsquigarrow[1]{%
    \stepcounter{sarrow}%
    \mathrel{\begin{tikzpicture}[baseline= {( $ (current bounding box.south) + (0,.2ex) $ )}]
    \node[inner sep=.5ex] (\thesarrow) {$\scriptstyle #1$};
    \path[draw,<-,decorate,
      decoration={zigzag,amplitude=0.7pt,segment length=1.2mm,pre=lineto,pre length=4pt}] 
        (\thesarrow.south east) -- (\thesarrow.south west);
    \end{tikzpicture}}%
    }
    
    \newcommand{\Real}{\mathbb{R}}

\newcommand{\thmheadfont}{\bfseries}
\newenvironment{repeatenv}[2]%
  {\smallskip\noindent {\thmheadfont #1~\ref{#2}.}\ \slshape}
  {\normalfont}

\colorlet{todo}{SeaGreen}

\title{Shortest Paths in Portalgons}
\author
{Maarten L\"{o}ffler}
{Department of Information and Computing Sciences, Utrecht University, the Netherlands
\and
Department of Computer Science, Tulane University, New Orleans, United States of America
\and \url{https://www.uu.nl/staff/MLoffler/Profile}}
{m.loffler@uu.nl}
{}
{}

\author
{Tim Ophelders}
{Department of Information and Computing Sciences, Utrecht University, the Netherlands 
\and 
Department of Mathematics and Computer Science, TU Eindhoven, the Netherlands}
{t.a.e.ophelders@uu.nl}
{}
{}

\author{Frank Staals}
{Department of Information and Computing Sciences, Utrecht University, the Netherlands}
{f.staals@uu.nl}
{}
{}

\author{Rodrigo I. Silveira}
{Department de Matem\`{a}tiques, Universitat Polit\`{e}cnica de Catalunya, Spain \and \url{http://dccg.upc.edu/people/rodrigo/}}
{rodrigo.silveira@upc.edu}
{}
{Partially funded by MICINN through project PID2019-104129GB-I00/ MCIN/ AEI/ 10.13039/501100011033.}

\authorrunning{L\"offler, Ophelders, Silveira, Staals}
\Copyright{Maarten L\"offler, Tim Ophelders, Rodrigo I. Silveira,
  Frank Staals}



\newcommand{\myremark}[4]{\textcolor{blue}{\textsc{#1 #2:}} \textcolor{#4}{\textsf{#3}}}

\newcommand{\frank}  [2][says]{\myremark{Frank}  {#1}{#2}{JungleGreen}}
\newcommand{\maarten}[2][says]{\myremark{Maarten}{#1}{#2}{WildStrawberry}}
\newcommand{\rodrigo}[2][says]{\myremark{Rodrigo}{#1}{#2}{BurntOrange}}

\renewcommand{\myremark}[4]{} 

\nolinenumbers
\begin{document}
\maketitle
    
\begin{abstract}
  Any surface that is intrinsically polyhedral can be represented by a collection of simple
  polygons (\emph{fragments}), glued along pairs of equally long oriented
  edges, where each fragment is endowed with the geodesic metric
  arising from its Euclidean metric.  We refer to such a
  representation as a \emph{portalgon}, and we call two portalgons equivalent
  if the surfaces they represent are isometric.

  We analyze the complexity of shortest paths.
  We call a fragment \emph{happy} if any shortest path
  on the portalgon visits it at most a constant number of times.  A
  portalgon is happy if all of its fragments are happy.  We present an
  efficient algorithm to compute shortest paths on happy
  portalgons.

  The number of times that a shortest path visits a
  fragment is unbounded in general.  We contrast this by showing that
  the intrinsic Delaunay triangulation of any polyhedral surface
  corresponds to a happy portalgon.  Since computing the intrinsic
  Delaunay triangulation may be inefficient, we provide an
  efficient algorithm to compute happy portalgons for a restricted
  class of portalgons.
\end{abstract}

\section{Introduction}

\begin{figure}[tb]
  \centering
  \includegraphics [width=\textwidth] {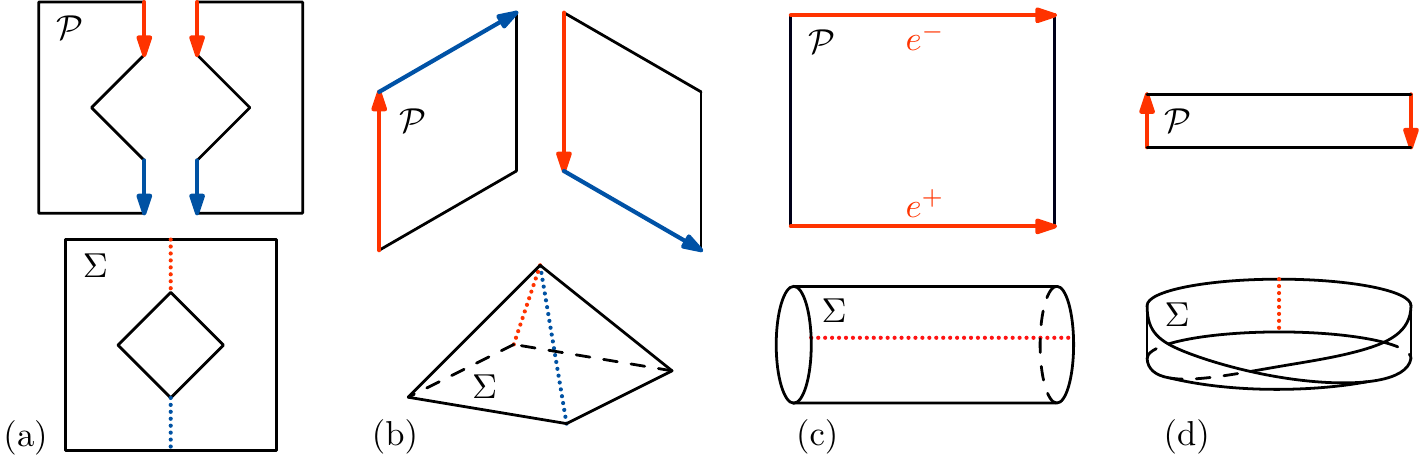}
  \caption{Examples of portalgons. Arrows of the same color represent portals that are identified with each other. (a) A portalgon that represents a polygon with a hole. (b) A portalgon that represents the surface of a bottomless pyramid. (c) A portalgon that represents the surface of a cylinder. (d) A portalgon that represents a M\"obius strip, a non-orientable surface.
  }
  \label{fig:intro_portalgons}
\end{figure}

  We define a {\em portalgon} $\P$ to be a collection of simple
  polygons ({\em fragments}) with some pairs of edges identified, see Figure~\ref {fig:intro_portalgons}.
  When we stitch together all fragments of a portalgon we obtain a two-dimensional  surface $\Sigma$, whose {\em intrinsic metric} is {\em polyhedral}~\cite {ybz-1993,milka68}.
  Note that $\Sigma$ is not necessarily embeddable in $\R^3$ with flat faces or without self-intersections, and not necessarily orientable.
  We say that $\P$ is a {\em representation} of $\Sigma$; crucially, the same surface may in principle have many different representations.
  Portalgons can be seen as a generalization of simple polygons, polygons with holes, polyhedral surfaces, and even developable surfaces.

We are interested in the {\em computational complexity}
of computing {\em shortest paths} on  portalgons.
In particular, we
analyze the {\em shortest path map} $\SPM(s)$; a representation of all
shortest paths from a source point $s$ to all other points in the
portalgon.
Our main insights are:
  \begin {itemize}
  \item The complexity of a shortest path 
    on a surface $\Sigma$, represented by a given portalgon \P, may be {\em
      unbounded} in terms of the combinatorial properties of $\Sigma$ and \P.
  \item The complexity of a shortest path depends on a parameter of the
    particular portalgon \P representing the surface that we refer to
    as its \emph{happiness} $h$. In particular, we show that the
    maximum complexity of a shortest path is $\Theta(n+hm)$, where $n$ is the
    total number of vertices in the portalgon, $m \leq n$ is the
    number of portals.

  \item Given a source point in \P, the complexity of its shortest
  path map is $O(n^2h)$.
  Moreover, if \P is triangulated, it can be computed in $O(\lambda_4(k)\log^2 k)$ time,
  where $k$ is the output complexity, and  $\lambda_4(k)$ the
  length of an order-4 Davenport-Schinzel sequence on $k$ symbols.

  \item Every surface with a polyhedral intrinsic metric  admits a
    representation as a portalgon where the happiness $h$ is
    constant. Specifically, one such representation is given by its intrinsic Delaunay triangulation. In such portalgons shortest paths have complexity
     $O(n)$.
  \item Since the intrinsic Delaunay triangulation is not easy to compute, we investigate the problem of transforming a given portalgon of happiness $h$ into one with constant happiness.
  We present an algorithm to do so in $O(n + \log h)$ time, for a restricted class of portalgons. The question of how to compute such a
    good representation, in general, remains open.
  \end{itemize}


%
\begin{figure}[tb]
  \centering
  \includegraphics [width=\textwidth] {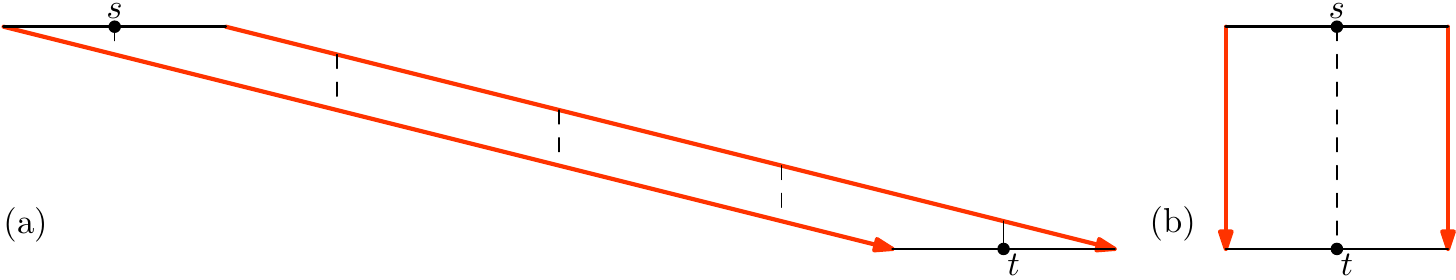}
  \caption{ (a) A portalgon where a shortest path between $s$ and $t$ (dashed) has unbounded complexity. (b) A different representation of the same surface where the path has constant complexity.
  }
  \label{fig:intro_complexity}
\end{figure}

  \subsection{Comparison to related work}

  Shortest paths have been studied in many different geometric
  settings, such as simple polygons, polygonal domains, terrains,
  surfaces, and polyhedra~(see
  e.g.~\cite{ChenH96,GuibasHershSPPoly,hershberger1999sssp,schreiber10optim_time_algor_short_paths_realis_polyh};
  refer to~\cite{MitchellSurvey17} for a comprehensive survey).
  The efficient \emph{computation} of shortest paths is a fundamental problem
  in computational
  geometry~\cite{Mitchell1991,mitchell91weigh_region_probl,guibas1987,ChenH96,wang21short_paths_among_obstac_plane_revis,wang21new_algor_euclid_short_paths_plane}. When
  the environment is a simple polygon the situation is well
  understood~\cite{guibas1987,GuibasHershSPPoly}. For polygons with
  holes, efficient solutions have also been known for quite a
  while~\cite{hershberger1999sssp}, recently culminating in an optimal
  $O(n + k\log k)$ time algorithm, where $n$ is the total number of
  vertices of the polygon, and $k$ the number of
  holes~\cite{wang21new_algor_euclid_short_paths_plane}. For more
  complex environments, like the surface of a convex polyhedron,
  several algorithms have been developed~\cite{ChenH96,mitchell87discr_geodes_probl,XinWang-09}, and even
  implemented~\cite{kaneva2000implementation,cgalchenhan}.
  However, for more
  general surfaces the situation is less well understood.

Portalgons generalize many of the geometric settings studied before. Hence, our goal is to unify existing shortest paths results. 
To the best of our knowledge, this 
  has not been attempted before, even though several questions closely related to the ones addressed in this work were posed as open problems in a blog post almost two decades ago~\cite{pancakes06}. Instead, portalgons are a rather
  unexplored concept which, though it has been long adopted into
  popular culture~\cite{portal1, startrek, matrix}, have only been
  studied from a computational point of view in the context of annular
  ray shooting by Erickson and
  Nayyeri~\cite{erickson13tracin_compr_curves_trian_surfac}.

  When measuring the complexity of a shortest
  path, 
  we can distinguish between its intrinsic complexity, and complexity
  caused by the representation of the underlying surface. For example, a shortest path $\geod$
  on a convex polyhedron $\Sigma$ in $\R^3$ with $n$ vertices, may
  cross (and thus bend) at $O(n)$ edges. Hence, a description of
  $\geod$ on $\Sigma$ has complexity $O(n)$. However, it is known that
  any convex polyhedron in $\R^3$ can be unfolded into a simple planar
  polygon $P_\Sigma$, and in such a way that a shortest path $\geod$
  corresponds to a line segment in
  $P_\Sigma$~\cite{agarwal97star_unfol_polyt_applic,aronov92nonov,ChenH96}. Hence,
  $\geod$ actually has a constant complexity description in
  $P_\Sigma$.  It is easy to see that some portalgons may have
  shortest paths 
  of {\em unbounded} complexity; as illustrated in
  Figure~\ref{fig:intro_complexity}(a).  This unbounded complexity,
  however, is completely caused by the representation, and indeed
  there is another equivalent portalgon without this behavior
  (Figure~\ref{fig:intro_complexity}(b)).  We introduce a parameter
  that explicitly measures the potential for shortest paths to have
  high complexity, which we call {\em happiness}---refer to
  Section~\ref {sec:defs} for a formal definition.

  In Section~\ref {sec:Shortest_Paths_in_Portalgons}, we analyze the
  complexity of shortest paths in terms of the happiness. Our first main
  result is that, if we have a portalgon with $n$ vertices and happiness $h$, then  the complexity of its shortest
  path map from a given source point is $O(n^2h)$.
  Moreover, we show that, for triangulated portalgons, it can be
  computed in an output-sensitive fashion: if $k$ is the complexity of the shortest path map, then it can be computed in $O(\lambda_4(k)\log^2 k)$ time.
  
  It is worth noting that our analysis of the shortest path map has similarities with techniques used to compute shortest paths on polyhedral surfaces, most notably~\cite{ChenH96,mitchell87discr_geodes_probl}.
  However, the fact that  portalgons are more general implies important differences with previous methods.
  We need to handle surfaces of non-zero genus (e.g.,  see Figure~\ref {fig:flat_torus}),
\begin {figure}
  \centering
  \includegraphics {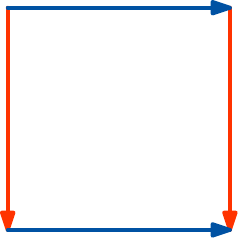}
  \qquad
  \includegraphics [height=10em] {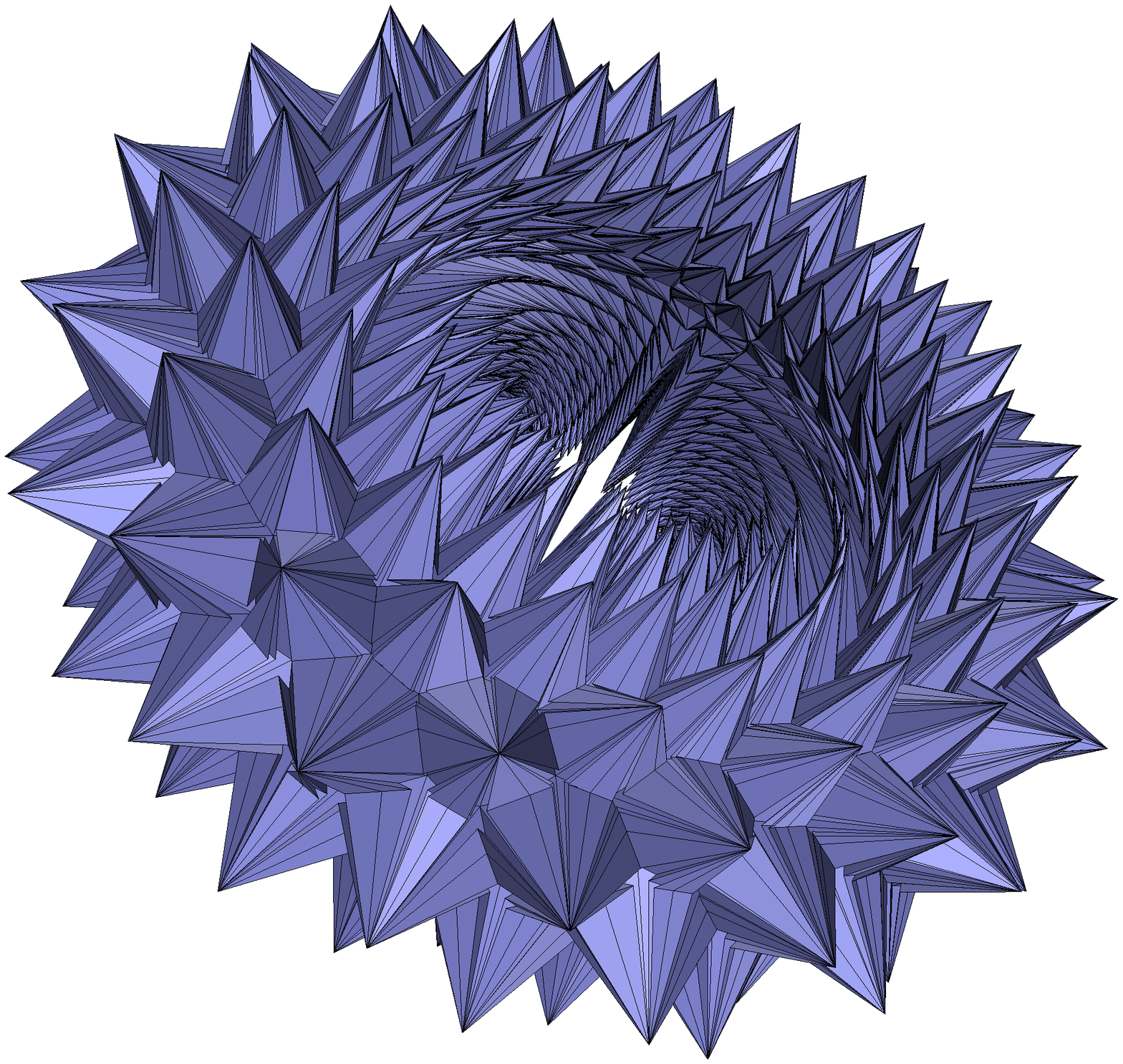}
  \caption {(a) A portalgon representing a flat torus. (b)
    Interestingly, a flat torus {\em can} be embedded isometrically in
    $\R^3$; however, the resulting embedding has very high
    complexity~\cite {borrelli12flat_tori,lazarus22flat_tori} (image
    from \url{https://www.imaginary.org/es/node/2375}).}
  \label {fig:flat_torus}
\end {figure}  
    while Chen and Han~\cite{ChenH96} require genus zero to compute the map in the interior of triangles
  (see the proof of their Theorem 4).
   Moreover, $\Sigma$ may be non-embeddable in Euclidean space with flat triangles.
  Another difference is that shortest paths in portalgons can cross the same portal edge multiple times, something that is often (explicitly or implicitly) assumed to be impossible in algorithms for polyhedral surfaces (e.g., in \cite{mitchell87discr_geodes_probl}).
  





The fact that the complexity of the shortest path map can be upper bounded by a function of the happiness, leads to the following natural question: for a given
  surface $\Sigma$, can it always be represented by a portalgon $\P$
  with {\em bounded happiness}?  In
  Section~\ref{sec:Existence_of_Happy_Portalgons}, we prove that the
  answer to this question is ``yes''.  In particular, our second main
  result is that for any portalgon, its {\em intrinsic Delaunay
    triangulation}~\cite {bobenko2007discrete} has constant happiness.
  \maarten
  {Discuss some relevant related work here?}

  In turn, this then leads to another natural question: given a
  surface $\Sigma$, can we actually efficiently {\em compute} a
  portalgon representing it that has bounded, preferably constant,
  happiness? Clearly, the answer to this question depends on how
  $\Sigma$ is represented. When $\Sigma$ is given as a portalgon \P,
  possibly with unbounded happiness, this question then corresponds to
  the problem of transforming \P into an equivalent portalgon with
  constant happiness. This is the problem we study in
  Section~\ref{sec:Making_Portalgons_Happy}. Given the above result,
  the natural approach is to try to construct the intrinstic Delaunay
  triangulation of \P. Unfortunately, it is unknown if it is possible
  to do that efficiently (that is, with guarantees in terms of $n$ and
  $h$). Our third main result shows that for a restricted class of
  portalgons we can give such guarantees. In particular, if the input
  portalgon has only one portal, $n$ vertices, and happiness $h$, we
  can construct an equivalent portalgon that has constant happiness in
  $O(n+\log h)$ time.

\section{Definitions and observations}
\label {sec:defs}


\subparagraph{Portalgons.}  We define a \emph{portalgon} \P to be a
pair $(\F, P)$,
where \F is a collection of simple polygons, called
\emph{fragments}, and $P$ is a collection of portals.  
A
\emph{portal} is an unordered pair $e=(\self e,\twin e)$ of directed,
distinct, equal
length, edges from some fragment(s) of \F. We refer to $\self e$ and
$\twin e$ as \emph{portal edges},
see Figure~\ref{fig:intro_portalgons}(c), and require that each portal
edge appears in one portal. If $\self p$ is a point on $\self e$, then
$\twin p$ will denote the corresponding point on $\twin e$.

Let $n$ be the total number of vertices in the fragments of \P, and
let $m$ be the number of portal edges; 
the number of portals is
$m/2$. Note that $m \leq n$. We denote the number of
vertices and the number of portal edges in a fragment $F \in \F$ by
$n_F$ and $m_F$, respectively.


If we ``glue'' the edges of the fragments along their common portal edges, then
the portalgon \P describes 
a surface (2-manifold with boundary) $\Sigma=\Sigma(\P)$, see
Figure~\ref{fig:intro_portalgons}. Specifically, $\Sigma$ is the space
obtained from \P by taking the collection \F and identifying
corresponding pairs of points on portal edges.
We can
write $\Sigma$ as a quotient space $(\bigcup_{F \in \F} F)/\sim$,
where $\sim$ is an equivalence relation that glues corresponding
portal edges $\self e$ and $\twin e$. 


\subparagraph{Fragment Graph.} A portalgon $\P=(\F,P)$ induces a
(multi)graph that we refer to as the \emph{fragment graph} $G$ of
\P (see Figure~\ref{fig:fragmentgraph}). Each fragment $F \in \F$ is a
node in $G$, and there is a link between $F_1$ and $F_2$ in
$G$ if and only if there is a portal $e$ with $\self e$ in
$F_1$ and $\twin e$ in $F_2$, or vice versa. Note there may be
multiple portals connecting $F_1$ and $F_2$. 
%
\begin{figure}[tb]
\centering
\includegraphics{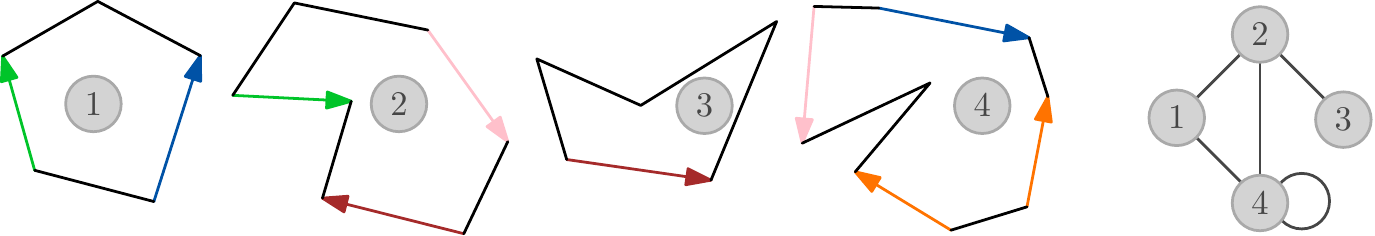}
\caption{The fragment graph of a portalgon with four fragments and five portals.
    }
\label{fig:fragmentgraph}
\end{figure}

\subparagraph{Paths and shortest paths.} A path $\pi$ from
$s \in \Sigma$ to $t \in \Sigma$ is a continuous function mapping the
interval $[0,1]$ to $\Sigma$, where $s=\pi(0)$ and $t=\pi(1)$. For two
points $p=\pi(a)$ and $q=\pi(b)$, we use $\pi[p,q]$ to denote the
restriction of $\pi$ to the interval $[a,b]$.
%
The fragments of \P split $\pi$ into a set $\Pi$ of maximal,
non-empty, subpaths $\pi_1,..,\pi_z$, where (the image of) each
$\pi_i$ is contained in a single fragment. To be precise, for each
fragment $F$, the intersection of (the image of) $\pi$ with $F$ is a
set of maximal subpaths, and $\Pi$ is the union of those sets over all
fragments $F$\footnote{Note that if $\pi$ passes through a vertex of
  \P that appears in multiple portals, $\Pi$ contains subpaths for
  which the image consists of only a single point; the vertex
  itself. We later restrict our attention to minimum complexity paths,
  which allows us to get rid of such singleton paths.  }.
  We define the length of $\pi$ as the sum of the lengths of its subpaths, 
  and the distance $d(s,t)$ between $s \in \Sigma$ and
$t \in \Sigma$ as the infimum of length over all paths between $s$ and
$t$. 
  We inherit the property that $d(s,t)=0$ if and
only if $s=t$ from the 
metric in each fragment. It
then follows that $d$ is also a metric. Moreover, $(\Sigma,d)$ is a
geodesic space.

Observe that if $\pi$ is a shortest path between $s$ and $t$, each
subpath $\pi_i$ is a polygonal path whose vertices are either
endpoints of $\pi_i$ or vertices of the fragment containing
$\pi_i$. Furthermore, when $\pi$ crosses a portal the path does not
bend (otherwise we can again locally shortcut it). It then follows
that a shortest path $\pi$ is also polygonal and can be uniquely
described by an alternating sequence $s=v_1,E_1,v_2,..,E_k,v_{k+1}=t$
of vertices ($s$, $t$, or portalgon vertices) and sequences of portal
edges crossed by the path. We refer to such a description as the
\emph{combinatorial representation} of the path, and to the total length of
these sequences as the \emph{complexity} of $\pi$. In the remainder of
the paper, we will use $\geod(s,t)$ to denote an arbitrary minimum
complexity shortest path between $s$ and $t$. As we observed before
(see Figure~\ref{fig:intro_complexity}(b)), a shortest path may still
intersect a single portal edge many times, and hence the complexity of
a shortest path may be unbounded in terms of $n$ and $m$.

\subparagraph{Isometry.}
A map $f\from X\to Y$ between metric spaces $X$ and $Y$ is an \emph{isometry} if $d_X(x,x')=d_Y(f(x),f(x'))$ for all $x,x'\in X$.
We say that $f$ is a \emph{local isometry} at a point $x$ if there exists an open neighborhood $U_x$ of $x$ such that the restriction of $f$ to $U_x$ is an isometry.

\subparagraph{Equivalent portalgons.}
Given a portalgon \P, there are many other portalgons that describe
the same surface $\Sigma$.
For instance, we can always cut
a fragment into two smaller fragments by transforming a chord of the
fragment into a portal, or, assuming this does not cause any overlap,
we can glue two fragments along a portal, see Figure~\ref{fig:equivalence}. 
More formally, two
portalgons $\P$ and $\mathcal{Q}$ are \emph{equivalent}, denoted $\P \equiv \mathcal{Q}$,
if there is a bijective isometry between them (i.e., if 
for any pair of points $s,t \in \P$, their distance in $\P$ and $\mathcal{Q}$ is the same).

\begin{figure}[tb]
  \centering
  \includegraphics{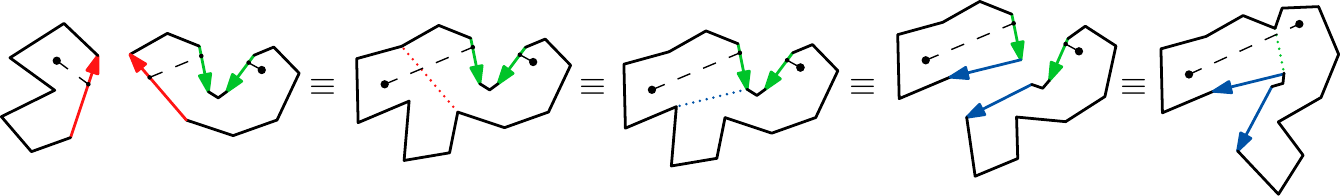}
  \caption{Five equivalent portalgons, with the shortest path between (the same) two points.}
  \label{fig:equivalence}
\end{figure}

\subparagraph{Happiness.}  Our ultimate goal will be to find, for a
given input portalgon, an equivalent portalgon such that all its
shortest paths have bounded complexity.  To this end, we introduce the
notion of a {\em happy portalgon}, and more specifically, a {\em happy
  fragment} of a portalgon.

Let $\P=(\F, P)$ be a portalgon, let $F \in \F$ be a fragment in \F,
and let $\Pi(p,q)$ denote the set of all 
shortest
paths between $p,q \in \Sigma$. We define $c(X)$ as the number of
connected components in $X$. The \emph{happiness}
$\H(F)=\max_{p,q\in\Sigma}\max_{\pi \in \Pi(p,q)} c(F\cap\pi)$ of
fragment $F$ is defined as the maximum number of times a
shortest path $\pi$ between any pair of points
$p,q \in \Sigma$ can go through the fragment.
 The happiness of \P is
then defined as $\H(\P)=\max_{F\in\F}\H(F)$ the maximum happiness over
all fragments. We say that a portalgon is $h$-happy when
$\H(\P) \leq h$.
Further, we will sometimes refer to a  \emph{happy} portalgon, without an $h$ value, to mean $O(1)$-happy.

\frank{I think here we would want some observation/claim/statement
  about *when* a segment $\overline{pq}$ is a shortest path.}

\begin{restatable}{lemma}{lemTriangulationPreservesHappiness}
  \label{lem:triangulation_preservers_happiness}
  Let $\P=(\F,P)$ be an $h$-happy portalgon, and let $\P'=(\F',P')$ be
  a triangulation of $\P$. The happiness of $\P'$ is at most $h$; that
  is, $\P'$ is an $h$-happy portalgon.
\end{restatable}

Note that if $\P$ has $n$ vertices and $m$ portals, any triangulation
$\P'$ of $\P$ consists of $n$ vertices and $O(m+n)$ portals. The fact
that in an $h$-happy portalgon a shortest path crosses every portal at
most $h-1$ times implies the following.

\begin{restatable}{lemma}{lemShortestPath}
  \label{lem:shortest_path}
  Let \P be an $h$-happy portalgon
  and let $s$ and $t$ be two points in \P. A shortest path
  $\geod(s,t)$ between $s$ and $t$ has complexity $O(n+hm)$. This
  bound is tight in the worst case.
\end{restatable}

\begin{proof} 
  The vertices of $\geod(s,t)$ are either: (i) $s$ or $t$ itself, (ii)
  vertices of \P, or (iii) points in which $\geod(s,t)$ crosses a
  portal edge. There are only two vertices of type (i) on
  $\geod(s,t)$. A shortest path can visit any vertex of \P at most once; hence, there are at
  most $n$ vertices of type (ii). Finally, since \P is $h$-happy,
  $\geod(s,t)$ crosses every fragment $F \in \F$ at most $h$ times.
  \begin{figure}[tb]
    \centering
    \includegraphics{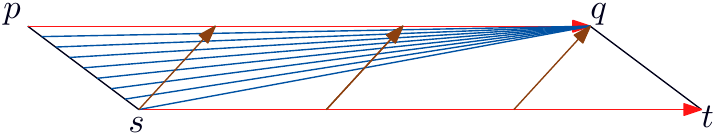}
    \caption{A simple $h$-happy portalgon in which a shortest path has
      complexity $\Omega(hm)$ as it crosses $\Omega(m)$ portal edges
      (the blue portal edges) $\Omega(h)$ times each.}
    \label{fig:lowerbound_complexity_path_main}
  \end{figure}
  A shortest path of complexity $\Omega(n)$ is easy to attain in a
  fragment without portals and a chain with $\Omega(n)$ reflex
  vertices. We get the $\Omega(hm)$ term using a portalgon like in
Figure~\ref{fig:lowerbound_complexity_path_main}. For any $h$, we can choose the
  length of the red portal so that the portalgon is $h$-happy.
\end{proof}

\subparagraph{Shortest path map.} Given a portalgon \P, a source point
$s \in \Sigma$, and a region $\X \subseteq \Sigma$ the shortest path
map $\SPM_\X(s)$ of $s$ is a subdivision of $\X$ into maximally
connected regions, such that for all points $q$ in the interior of a
region $R \in \SPM_\X(s)$ the shortest path from~$s$ to $q$ is unique,
and has the same combinatorial structure, i.e., visits the same
sequence of vertices and portal edges of \P. Note that the complexity
of $\SPM_\X(s)$ depends on the representation of the surface
$\Sigma=\Sigma(\P)$, that is, the portalgon \P. So changes to \P may
affect the complexity of $\SPM_\X(s)$. For example, splitting faces of
\P increases the complexity of $\SPM_\X(s)$. Hence, when $\Sigma$ is
fixed, an important problem is to find a good portalgon (i.e. one for
which $\SPM(s)=\SPM_\Sigma(s)$ has low complexity) representing it.


\subparagraph{Intrinsic Delaunay triangulation.}
A triangulation of a portalgon is an equivalent portalgon whose vertex set is the same, and all of whose fragments are triangles.
In particular, among all such triangulations, the intrinsic Delaunay triangulation is such that for any interior edge of the triangulation, for the two triangles $t$ and $t'$ incident to that edge, the corners of $t$ and $t'$ not incident to that edge sum up to at most 180 degrees~\cite{bobenko2007discrete}.


\section{Shortest paths in portalgons}
\label{sec:Shortest_Paths_in_Portalgons}

In this section we show how to compute the shortest path map
$\SPM_\P := \SPM_\P(s)$ of a source point $s$ in a $h$-happy
triangulated portalgon \P. By
Lemma~\ref{lem:triangulation_preservers_happiness} any $h$-happy
portalgon can be transformed into a triangulated $h$-happy portalgon
by triangulating its fragments and replacing the resulting diagonals
by portals. For ease of description, we assume that any two triangles
have at most one edge in common, and that $s$ is a vertex of the
triangulation. We can again subdivide every triangle into $O(1)$
$h$-happy triangles to achieve this.

For a path $\pi$ starting from the source $s$, define its signature, denoted $\sigma(\pi)$, to be the sequence of vertices and portals it passes through.
Note that paths may simultaneously pass through a vertex and a portal (or multiple portals incident to that vertex).
In this case, we break ties in the sequence by placing vertices before portals, and portals in the order that the path would pass through them if it were perturbed away from vertices (in a consistent way), where we always perturb the start of the path into a fixed triangle $T_s$ incident to $s$.

If $\pi$ is a shortest path from the source $s$ to a point $p$ in a triangle $T$, and $v$ is the last vertex on $\sigma(\pi)$, then the length of $\pi$ is $d(s,v)+d(v,p)$, where $d(v,p)$ can be expressed as the length of a line segment as follows.
We think of $T$ as being embedded locally isometrically in the Euclidean plane, and by unfolding the fragments that $\pi$ passes through after $v$, we can compute a copy of $v$ in the plane, as well as copies of the portals and triangles that $\pi$ passes through after~$v$; see Figure~\ref{fig-full:portalSequenceWindow}.
The locations of these copies in the plane depend only on $\sigma(\pi)$ (but not $p$ or $\pi$ itself).
Then, $d(v,p)$ is the Euclidean distance between the copies of $v$ and $p$ in the plane, and the segment connecting these copies passes through all the unfolded copies of portals of $\sigma(\pi)$ after $v$.
Let $e$ be an edge of $T$, and define $I_{\sigma(\pi),e}$ to be the interval of points $p$ on $e$ for which the segment $\overline{vp}$ passes through all the unfolded copies of portals of $\sigma(\pi)$ after $v$.
For a point $q$ in $T$, define
\[
f_{\sigma(\pi)}(q)=\begin{cases}
    d(s,v)+\|\overline{vq}\| & \text{if the segment $\overline{vq}$ passes through $I_{\sigma(\pi),e}$,}\\
    \infty & \text{otherwise.}
\end{cases}
\]
Define $d_{\sigma(\pi)}(q)$ to be the infimum length over paths from $s$ to $q$ with signature $\sigma(\pi)$.
This infimum is not necessarily realized by a path with the same signature, but is always realized by a path that potentially bends around additional vertices after the last vertex of $\sigma(\pi)$, which are therefore inserted in its signature.
We say that such a signature \emph{reduces} to $\sigma(\pi)$.
Note that if $f_{\sigma(\pi)}(q)$ is finite, then $f_{\sigma(\pi)}(q)=d_{\sigma(\pi)}(q)$.
If $\pi$ is a shortest path from $s$ to $p$, then $\pi$ has length $f_{\sigma(\pi)}(p)$.
For a portal $e$ of $T$, let $f_{\sigma(\pi)|e}\from e\to\R\cup\{\infty\}$ be the restriction of $f_{\sigma(\pi)}$ to points of $e$.

    \begin{figure}[tb]
        \centering
        \includegraphics{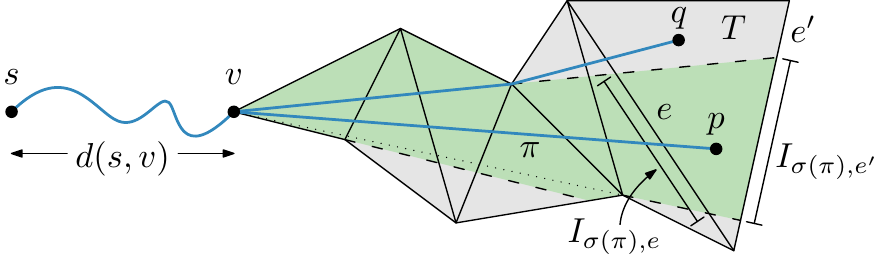}
        \caption{The length of the path $\pi$ from $s$ to point $p$ in
          triangle $T$ is $d(s,v)+\|p-v\|$. The signature of $\pi$
          defines intervals $I_{\sigma(\pi),e}$ and
          $I_{\sigma(\pi),e'}$ on edges $e$ and $e'$ of $T$, as well
          as a distance function $d_{\sigma(\pi)}$ illustrated by the
          path to point $q$ in $T$.}
        \label{fig-full:portalSequenceWindow}
    \end{figure}

\subsection{A data structure for maintaining a lower envelope}
\label{app:A_data_structure_for_maintaining_a_lower_envelope}

Let $F$ be a set of $m$ partial continuous functions, each defined on
a single interval, and such that each pair of functions can intersect
at most twice. We describe a data structure storing the lower envelope
$\env_F$ of $F$, that supports the following operations:

\begin{description}
\item[\textsc{NextLocalMinimum}($\delta$):] report the smallest local
  minimum of $\env_F$ that is larger than $\delta$.
\item[\textsc{NextVertex}$(f,q)$:] given a function $f \in F$ that
  realizes $\env_F$ at point $q$, find (if it exists) the lowest
  endpoint $(v,\delta')$ of the segment\maarten {What is a segment? Since the functions are "partial unimodal", are they composed of "segments", or are we looking for a vertex caused by the intersection of two different functions here?}\frank{in our application it will always be an intersection point; but it may also just be the endpoint of the interval on which $f$ is defined; that shouldn't be a problem}
  of $\env_F$ containing
  $(q,f(q))$ for which $\delta' > f(q)$.
\item[\textsc{Insert}($f$):] insert a new function into $F$.
\end{description}



We develop a simple static data structure supporting the
\algcall{NextLocalMinimum} and \algcall{NextVertex} operations, that
we turn into an efficient insertion-only data structure using the
logarithmic method~\cite{bentley1980decomposable}. In our
implementation of the static structure, the two query operations may
actually make modifications to the static structure that will
guarantee efficient amortized query times.

We represent $\env_F$ using $O(\log m)$ static data structures. Each
such static data structure stores the lower envelope $\env_i$ of a
subset $F_i$ of $2^i$ functions from $F$. In particular, it stores:
\begin{enumerate}[label=(\roman*)]
\item a binary search tree of (a subset of) the intervals of $\env_i$,
ordered by $x$-coordinate of their left endpoint, and \item a sorted
list of local minima on $\env_i$ ordered on increasing $y$-coordinate.
\end{enumerate}
We can now easily answer \algcall{NextLocalMinimum} queries in
$O(\log^2 m)$ time: for each of the $O(\log m)$ stored lower
envelopes, we binary search in the list of local minima to find the
smallest minimum with value larger than $\delta$, and report the
smallest one overall.


\begin{figure}[tb]
  \centering
  \includegraphics{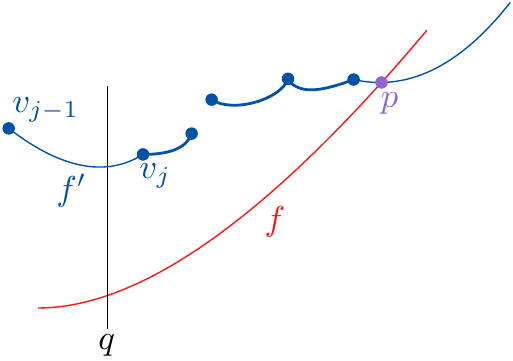}
  \caption{In a \algcall{NextVertex}$(f,q)$ query we find $\env_i(q)$
    on each envelope $\env_i$ (shown in dark blue), and walk along
    $\env_i$, to find the first intersection point $p$ of $f$ with
    $\env_i$. The fat blue segments of $\env_i$ are, and will forever
    be, dominated by $f$, hence we delete them.  }
  \label{fig:walk_envelope}
\end{figure}

To answer \algcall{NextVertex}$(f,q)$ queries we essentially compute
the endpoints of the segment of $\env_F$ containing $(q,\delta)$, and
report the lowest such endpoint with value larger than $\delta$. To
find the right endpoint of this segment, we do the following (finding
the left endpoint is symmetric). For each of the $O(\log m)$ static
structures, we do a binary search on the vertices of $\env_i$ to find
the interval $[v_{j-1},v_j]$ containing point $q$, and the function
$f'$ that realizes $\env_i$ in this interval. We then find the first
intersection point $p$ of $f$ with $\env_i$ to the right of $q$ by
walking along $\env_i$ (see Figure~\ref{fig:walk_envelope}). In
particular, we compute the intersection points of $f$ and $f'$, and if
there is such a point in the interval $[q,v_j]$, we found point
$p$. Otherwise, we continue the search (walk) with the next interval
$[v_j,v_{j+1}]$. We can stop the walk if we arrive at an interval
whose left endpoint lies strictly below $f$ (in that case there is no
intersection on $\env_i$). All intervals $[v_{j-1},v_j]$ that we
encounter during this walk and in which $f$ lies below $f'$ in the
entire interval will never show up on $\env_F$ (since $f$ will never
be deleted). Hence, we will delete these intervals from $\env_i$. This
allows us to bound the overall time spent walking along $\env_i$
envelopes over all queries. The running time of our query is thus
$O(\sum_{i=0}^{O(\log m)} |\env_i| + \ell) = O(\log m\log |\env_F| +
\ell) = O(\log^2 m + \ell)$, where $\ell$ is the total number of
intervals deleted during this walk.

Finally, to handle an \algcall{Insert}$(f)$ operation; we create a new
static data structure lower envelope structure $\env_0$ representing
just the singleton set $F'_0=\{f\}$. 
If there are ever two static data
structures representing sets $F_i$ and $F'_i$ of the same size $2^i$,
we replace them by a data structure representing a
new set $F'_{i+1}=F_i \cup F'_i$ of size $2^{i+1}$. We can construct this data structure by
merging the two earlier structures: i.e., we can simultaneously scan
through $\env_i$ and $\env'_i$ creating the new combined envelope in
time $O(|\env_i|+|\env'_i|+|\env'_{i+1}|) = O(|\env'_{i+1}|)$. During
this merge, we can mark the local minima that remain local minima. We
can then merge the sorted lists of marked local minima in linear
time.

Next, we prove that the total number of intervals in the lower
envelopes that we create during a sequence of $m$ insertions is at
most $O(|\env_F|\log m)=O(\lambda_4(m)\log m)$, where $\lambda_4(n)$
is the maximum length of a Davenport-Schinzel sequence of order four
on $n$ symbols. Hence, the amortized insertion time is
$O((\lambda_4(m)/m)\log m)$.
\begin{lemma}
  \label{lem:total_size_created}
  The total number of intervals in the $\env_i$ envelopes created over
  a sequence of $m$ insertions into $F$ is $O(\lambda_4(m)\log m)$.
\end{lemma}

\begin{proof}
  Consider the set of functions $F_j$ in the final structure, we bound
  all intervals created due to functions in $F_j$ by
  $O(\lambda_4(2^j)\log m)$. It then follows that the total number of
  intervals created is bounded by
  $\sum_{j=0}^{O(\log m)} O(\lambda_4(2^j)\log m) =
  O(\lambda_4(2^{O(\log m)})\log m) = O(\lambda_4(m)\log m)$ as
  claimed.

  Clearly, there is at most one interval created due to $F_0$. For all
  intervals created by functions in $F_j$, with $j > 0$, their
  intervals were created when we merged $\env_i$ and $\env'_i$,
  $0 \leq i < j$, into an $\env_{i+1}$. Charge the cost of creating
  the intervals from $\env_i$ and $\env'_i$ to the intervals of
  $\env_{i+1}$. It now follows that in the end, all costs are charged
  to intervals of $\env_j$. Moreover, every interval in $\env_j$ is
  charged at most $O(j)=O(\log m)$ times. Since $\env_j$ consists of
  at most $O(\lambda_4(2^j))$ intervals the lemma follows.
  \frank{not really happy with this proof yet. This should really be
    fairly standard though, but the $\lambda_4()$'s are somewhat annoying.}
\end{proof}

Since any interval can be deleted at most once, we thus obtain the following result:

\begin{restatable}{lemma}{lemLowerEnvelopeStructure}
  \label{lem:lower_envelope_structure}
  We can maintain the lower envelope $\env_F$ of a set $F$ of partial
  functions $f$, each pair of which intersect each other at most
  twice, in a data structure that: any sequence of $m$
  \algcall{Insert} operations and $k$ \algcall{NextLocalMinimum} and
  \algcall{NextVertex} queries takes at most
  $O(k\log^2 m + \lambda_4(m)\log m)$ time.
\end{restatable}

\subsection{SPM restricted to the edges}
\label{app:SPM_restricted_to_the_edges}

We first compute the shortest path map of $s$, restricted to the edges
of the portalgon. As in earlier
work~\cite{mitchell87discr_geodes_probl,mitchell91weigh_region_probl}
we propagate a wavefront of points at distance $\delta$ from $s$ as we
vary $\delta$ from
$0$ to $\infty$. However, it will be more convenient to view this as a
collection of simultaneous sweepline algorithms. For each portal edge
$e$, we sweep a horizontal line at height $\delta$ upward through the
``position along $e$ $\times$ distance from $s$'' space (see
Figure~\ref{fig-full:spm_data}), while we construct the part of
$\SPM_e$ below the sweep line. The main challenge is computing the
next event --the first vertex in the lower envelope of the distance
functions above the sweep line-- in time.

For each portal $e$ connecting two fragments $T_A$ and $T_B$, we
maintain the following information:
\begin{figure}[tb]
  \centering
  \includegraphics{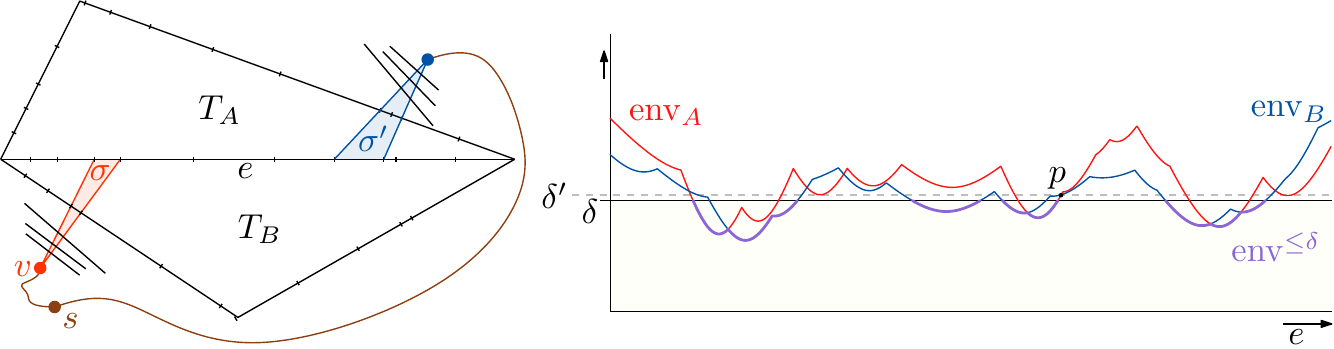}
  \caption{The information that we maintain while computing the SPM of
    the edges.}
  \label{fig-full:spm_data}
\end{figure}
\begin{enumerate}
    \item Let $S_A(e,\delta)$ (resp. $S_B(e,\delta)$) be the set of
      signatures of shortest paths\footnote{For efficiency, we
        symbolically ensure uniqueness of shortest paths by breaking
        ties consistently based on their signatures.
      } from $s$ to
      points on the boundary of $T_A$ (resp. $T_B$),
      where the last element of the
      signature is not $e$, and the length of the path is at most
      $\delta$.
      We represent each signature $\sigma\in S_A(e,\delta)\cup S_B(e,\delta)$ implicitly by storing the interval $I_{\sigma,e}$, the position of (the unfolded copy of) the last vertex $v$ on $\sigma$, and $d(s,v)$, so that we can compute $f_{\sigma|e}$ in constant time.
      For the purpose of recovering the shortest path map, we additionally store a pointer to the corresponding data on the previous edge on $\sigma$ (if any).
    \item We store the lower envelope $\env_A(e,\delta)$ of the
      functions $f_{\sigma|e}$, where $\sigma$ ranges over the
      signatures $S_A(e,\delta)$, in the data structure of Lemma~\ref{lem:lower_envelope_structure}.   Based on $S_B$, we symmetrically define and store $\env_B(e,\delta)$.
    \item
    Let $\envR(e,\delta)$ be the lower envelope of the functions $\env_A(e,\delta)$ and $\env_B(e,\delta)$.
    We maintain only the part of $\envR(e,\delta)$ that lies below the sweep line, denoted $\envM(e,\delta)$.

    Intuitively, $\envM(e,\delta)$ will be the part of $\SPM_e$ within distance $\delta$ from $s$.
    We store it in a balanced binary search tree of breakpoints along the edge.
    Note that some of these break points may lie on the sweep line, and thus move continuously with the sweep line.

  \item We maintain a binary search tree $\envR^{=\delta}(e,\delta)$
    storing the intersection points of $\env_A(e,\delta)$ and
    $\env_B(e,\delta)$ with the sweep line, in order along the sweep
    line. For each such intersection point we store the function(s)
    from $\env_A(e,\delta)$ or $\env_B(e,\delta)$ realizing this intersection point.

    \item Finally, we maintain a set of events pertaining to the edge $e$.
    We aggregate the events of all edges in a global priority queue, and use it to advance the sweep line algorithm to the next relevant value.
    The events that we store for an edge $e$ are the values $\delta'>\delta$ such that
    \begin{enumerate}
        \item $\delta'$ corresponds to a minimum of a function $f_{\sigma|e}$ with $\sigma\in S_A(e,\delta)\cup S_B(e,\delta)$,

        \item $\delta'$ corresponds to a vertex of~$\env_A(e,\delta)$
          or $\env_B(e,\delta)$, in particular an endpoint of one of
          the edges that currently appear in
          $\env^{=\delta}(e,\delta)$, or

        \item $\delta'$ corresponds to an intersection between
          functions $f_{\sigma|e}$ on $\env_A(e,\delta)$ and
          $f_{\sigma'|e}$ on $\env_B(e,\delta)$, where
          $f_{\sigma|e}$ and $f_{\sigma'|e}$ are neighbors in $\env^{=\delta}(e,\delta)$.
    \end{enumerate}
    For each event, we also keep track of the type and corresponding functions.
\end{enumerate}

Note that $\envR^{\leq\infty}(e,\infty)=\envR(e,\infty)$.
Before we show how to maintain all the above information, we show that this information indeed corresponds to the shortest path map.
\begin{lemma}
\label{lem:envIsSPM}
$\envM(e,\delta)$ encodes the shortest paths of length $\leq\delta$ to points on the edge $e$.
\end{lemma}
\begin{proof}
For any point $q \in e$ the (due to tie-breaking, unique) shortest path $\pi$ from $s$ to $q$ arrives from either $T_A$ or $T_B$.
If $\pi$ has length at most $\delta$, then $\sigma(\pi)$ appears in either $S_A(e,\delta)$ or $S_B(e,\delta)$, and hence $f_{\sigma(\pi)|e}(q)$ lies on or above the lower envelope $\envM(e,\delta)$.
Conversely, any point on $\envM(e,\delta)$ with finite function value corresponds to a path of that length, so $f_{\sigma(\pi)|e}(q)$ indeed lies on $\envM(e,\delta)$.
\end{proof}

For $\delta=0$ and any edge $e$, we initialize $S_A(e,0)$ to be empty if $T_A$ does not contain~$s$, and to consist of a single sequence consisting of $s$ and the edges around $s$ between $T_s$ and $T_A$ (excluding $e$) otherwise.
We initialize $S_B$ symmetrically.
Based on this, $\env_A$ and $\env_B$ are initially lower envelopes each of at most one function, where $\envM(e,\delta)$ is empty unless $s\in e$.
The priority queue of events related to $e$ therefore consists of at most one event: $\nextE(e,0)$ is the distance from~$s$ to $e$ inside $T_A$ or $T_B$ (if $s$ lies in $T_A$ or $T_B$), and $\infty$ otherwise.
It remains to argue that we can correctly maintain the information as $\delta$ increases.
\begin{figure}[tb]
  \centering
  \includegraphics{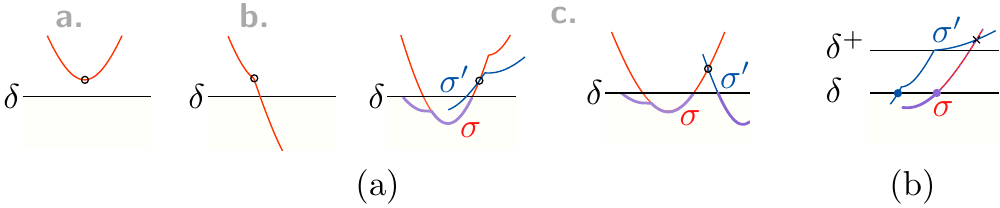}
  \caption{(a) The different event types: \enumit{a.} local minima in
    $\env_A$ or $\env_B$, \enumit{b.} breakpoints of $\env_A$ or
    $\env_B$, and two events of type \enumit{c.} first intersections
    of $\env_A$ and $\env_B$. (b) We will detect the intersection
    between $\env_A$ and $\env_B$ once $f_{\sigma|e}$ and
    $f_{\sigma'|e}$ become neighbors in $\env^{=\delta}$ at time
    $\delta^+$.  }
  \label{fig:events}
\end{figure}

Observe, that our events allow us to correctly detect the first time
$\delta' > \delta$ at which $\envR^{\leq \delta'}(e,\delta)$ differs
combinatorially from $\envR^{\leq \delta}(e,\delta)$, provided that
$S_A(e,\delta)$ and $S_B(e,\delta)$ remain unchanged. In particular,
any local minimum $\envR(e,\delta)$ is a local minimum of
$\env_A(e,\delta)$ or $\env_B(e,\delta)$, and any vertex of
$\envR(e,\delta)$ is either a vertex of $\env_A(e,\delta)$ or
$\env_B(e,\delta)$, or an intersection point of $\env_A(e,\delta)$
with $\env_B(e,\delta)$. The functions $\env_A(e,\delta)$ or
$\env_B(e,\delta)$ can only intersect at time $\delta'$ when
$f_{\sigma|e}(e,\delta)$ and $f_{\sigma'|e}(e,\delta)$ are neighbors
on the sweep line at some time $\delta < \delta'$.

\subparagraph{Event handling.} At an event (of any type), the
order in which the functions of $\env_A(e,\delta)$ and
$\env_B(e,\delta)$ intersect the sweep line changes. We therefore update
$\env^{=\delta}(e,\delta)$ by removing and inserting the appropriate
functions associated with this event, and additionally make sure that
we discover any additional events.

For each newly inserted function $f_{\sigma|e} \in \env_A(e,\delta)$
that intersects the sweep line in a point $(q,\delta)$, we use a
\algcall{NextVertex}$(f_{\sigma|e},q)$ query to find the next
type \enumit{b.} event where the sweep line will pass over a vertex of
$\env_A(e,\delta)$. We use an analogous query for any function
$f_{\sigma'|e} \in \env_B(e,\delta)$.

Furthermore, for any function $f_{\sigma|e} \in \env_A(e,\delta)$ that
has a new neighbor $f_{\sigma'|e} \in \env_B(e,\delta)$ in the order
along the sweep line, we compute if these functions intersect above
the sweep line, and if so create a new type \enumit{c.} event in the
event queue. We handle the case where
$f_{\sigma|e} \in \env_B(e,\delta)$ analogously. Furthermore, we
remove any type \enumit{c.} events of functions that are no longer
neighbors.

If this event was a local minimum of $\env_A(e,\delta)$ we extract the
next local minimum of $\env_A(e,\delta)$ above the sweep line using a
call to \algcall{NextLocalMinimum} on $\env_A(e,\delta)$. We handle
local minima of $\env_B(e,\delta)$ analogously.

Finally, as a result of the event, a new function, say $f_{\sigma'|e}$
from $\env_B(e,\delta)$, may have appeared on
$\envM(e,\delta)$. Hence, we insert it into $\envM(e,\delta)$, and
propagate $\sigma'$, extended by edge $e$, into the sets
$S_B(e',\delta)$ of the other two edges $e'$ incident to $T_A$. We
therefore use an \algcall{Insert} call to insert a new function into
$\env_A(e',\delta)$ (if it is not already present), and
\algcall{NextLocalMinimum}$(\delta)$ call on $\env_A(e',\delta)$ to
make sure we update the next local minimum.

\frank{additional starts here}
When we insert $f_{\sigma'|e'}$ into $\env_A(e',\delta)$ it may
already intersect the sweepline in the $e' \times \mathit{distance}$
space; i.e. there may be an interval of points $q' \in [a',b']$ on
$e'$ for which $f_{\sigma'|e'}(q') \leq \delta$. We therefore insert
the at most two intersection points of $f_{\sigma'|e'}$ with the
sweep-line ($a'$ and $b'$) into $\env_=(e',\delta)$, and in turn
update the at most eight type \enumit{b.} and \enumit{c.} events (four
each) caused by these new intersection points.

Note however, that our invariants guarantee that any point $q'$ in the
interval $[a',b']$ on which $f_{\sigma'|e'}(q') \leq \delta$ we
actually have that $\envM(e',\delta)(q') \leq f_{\sigma'|e'}(q')$,
hence $\envM(e',\delta)$ remains unchanged (and hence no further
updates are triggered).

\begin{lemma}
  \label{lem:nocascades}
  For any point $q' \in [a',b'] \subseteq e'$ we have that
  $\envM(e',\delta)(q') \leq f_{\sigma'|e'}(q')$.
\end{lemma}

\begin{proof}
  Assume by contradiction that there is a point $q' \in [a',b']$ for
  which $f_{\sigma'|e'}(q') < \envM(e',\delta)(q')$. Moreover, let
  $q'$ be the closest such point over all edges in the portalgon;
  i.e. so that $f_{\sigma'|e'}(q')=d(s,q')=\delta_1 \leq \delta$ (and
  thus for any $\delta_0 < \delta_1$ the function $\envM(e,\delta_0)$
  correctly represents the shortest path map on (any) edge $e$). Let
  $q$ be the point where the shortest path $\pi(\sigma')$ to $q'$
  intersects $e$; this subpath $\pi_0$ thus has length
  $f_{\sigma'|e}(q) = \delta_0 < \delta_1$. By subpath optimality,
  $\pi_0$ is a shortest path, and thus $\delta_0=d(s,q)$. We then have
  that $\envM(e,\delta_0)=f_{\sigma'|e}(q)$. However, it then follows
  that $\sigma'$ was already inserted into
  $S_A(e',\delta_0) \subseteq S_A(e',\delta)$ at ``time''
  $\delta_0$. Therefore,
  $\envM(e',\delta)(q') \leq \env_A(e',\delta)(q') \leq
  f_{\sigma'|e'}(q')$, and thus we obtain a contradiction.
\end{proof}





\frank{and ends here}

\begin{remark}
  Mitchell, Mount, and
  Papadimitriou~\cite{mitchell87discr_geodes_probl} use a similar
  overall algorithm. They define a notion of \emph{$T_A$-free} paths
  that arrive at edge $e$ from $T_B$. They prove that these paths act
  sufficiently like ``real'' shortest paths so that they can
  explicitly maintain the set of shortest $T_A$-free and $T_B$-free
  paths. Unfortunately, some of the arguments used hold only when a
  shortest path may cross a portal edge at most once (that is, when
  the portalgon is $1$-happy). In case of the weighted region problem,
  Mitchel and Papadimitriou show how to deal with this by extending
  this notion of \emph{$T_A$-free} paths to \emph{locally}
  \emph{$T_A$-free}
  paths~\cite{mitchell91weigh_region_probl}. However, it is unclear
  how to bound the number of such paths when the genus may be
  non-zero.

  In contrast, we insert a signature $\sigma$ into $S_B(e,\delta)$ \emph{only} when
  we can guarantee that $\sigma$ produces a shortest path to a point
  on edge of $T_B$ (rather than already when it produces a locally
  shortest $T_B$-free path to such an edge). We can thereby limit the
  size of the sets $S_B(e,\delta)$, by charging them to
  intervals in the actual shortest path map.
\end{remark}

\subparagraph{Analysis.} Let $k=|\SPM_{\partial \P}|$ denote the
complexity of the shortest path restricted to the edges. We first show
that the total number of events, over all edges of \P is
$O(\lambda_4(k))$. We first count the number of events for a fixed
edge $e$.

Observe that for any $\delta$,
$S_A(e,\delta) \subseteq S_A(e,\infty)$, hence every event of type
\enumit{a.} is a minimum of a function $f_{\sigma|e}$ with
$\sigma \in S_A(e,\infty) \cup S_B(e,\infty)$. As each function is
unimodal it contributes at most one minimum, and hence the total
number of events of type \enumit{a.} is $O(k_A + k_B)$, where
$k_A = |S_A(e,\infty)|$ and $k_B=|S_B(e,\infty)|$.

\begin{lemma}
  \label{lem:number_of_vertices_created}
  The total number of events of type \enumit{b.} corresponding to
  vertices in $\env_A(e,\delta)$, over all values $\delta$, is
  $O(\lambda_4(k_A))$.
\end{lemma}

\begin{proof}
  For each signature $\sigma \in S_A(e,\infty)$, consider the time
  $\delta$ at which $\sigma$ was added to $S_A(e,\delta)$. Now
  restrict each function $f_{\sigma|e}(q)$ to the domain on which
  $f_{\sigma|e}(q) \geq \delta$. This essentially splits
  $f_{\sigma|e}$ into two functions. Let $F'$ denote the resulting set
  of $O(k_A)$ functions, and let $\env_{F'}$ be its lower envelope. It
  follows that this lower envelope has complexity
  $O(\lambda_4(k_a))$~\cite{sharir1995davenport}.

  At every event \enumit{b.} corresponding to $\env_A(e,\delta)$, for
  some value $\delta$, the sweepline passes through a vertex
  $v=(q,\delta)$ of $\env_A(e,\delta)$. We claim that $v$ is also a
  vertex $\env_{F'}$. As the sweepline passes through such a vertex at
  most once, the number of such events is thus at most
  $O(\lambda_4(k_a))$.

  Since $v$ is a vertex of $\env_A(e,\delta)$, it is an endpoint of
  (the graph of) a function in $F'$. Next, we argue that
  $\env_A(e,\delta)(q) \leq \env_{F'}(q)$, and thus $v$ must also lie
  on $\env_{F'}$.

  For any time $\delta$, and any point $q$ we have
  $\env_A(e,\delta) \leq \env_{F'}(q)$: Assume by contradiction that
  $f'_{\sigma|e}(q)=\env_{F'}(q) = \delta_2 < \env_A(e,\delta)$. So
  $\sigma$ was added to $S_A(e,\infty)$ at some time
  $\delta_1 \leq \delta_2$. However, then $\sigma \in S(e,\delta)$,
  and thus $\env_A(e,\delta) \leq f_{\sigma|e}(q) =
  \delta_2$. Contradiction. We thus have $\env_A(e,\delta) \leq \env_{F'}(q)$, which completes the proof.
\end{proof}

By Lemma~\ref{lem:number_of_vertices_created} (and a symmetric counterpart for $\env_B$), the number of events of type \enumit{b.} is
at most $O(\lambda_4(k_A) +\lambda_4(k_B))$. Every event of type \enumit{c.}  produces a vertex
of $\envM(e,\delta)$, and thus of $\SPM_e$. Hence, it follows that the
total number of events on $e$ is
$O(\lambda_4(k_A)+\lambda_4(k_B))+|\SPM_e|)$.

A signature $\sigma$ appears in $S_B(e,\delta)$ only when there
is a shortest path to another edge $e'$ of $T_B$. Since every such edge $e'$ propagates to
at most four edges (of the two triangles incident to $e'$) it then
follows that
$\sum_{e \in \partial\P} (|S_A(e,\delta)| +
|S_B(e,\delta)|)=O(|\SPM_{\partial\P}|)$. Therefore, the total
number of events, over all edges $e$, is $O(\lambda_4(k))$ as
claimed.

For every edge $e$, and at any time $\delta$, the size of
$\env^{=\delta}(e,\delta)$ is at most
$O(|\env_A(e,\delta)|+|\env_B(e,\delta)|)$, as every function
intersects the sweep line at most twice. Furthermore, as we argued
above, the size of all those functions $\env_A(e,\delta)$,
$\env_B(e,\delta)$, $\envM(e,\delta)$, over all edges $e$, is at most
$O(\lambda_4(k))$.

To handle an event on edge $e$, we use a constant number of update
operations on the event queue(s), $\envM(e,\delta)$, and
$\env^{=\delta}(e,\delta)$, and a constant number of calls to
\algcall{Insert}, \algcall{NextLocalMinimum}, and \algcall{NextEvent}
on $\env_A(e,\delta)$ and $\env_B(e,\delta)$. It follows the former
operations take only $O(\log k)$ time, and by
Lemma~\ref{lem:lower_envelope_structure}, the total cost of the
operations on $\env_A(e,\delta)$ and $\env_B(e,\delta)$, over all
edges, is $O(\lambda_4(k)\log^2 k)$. It follows that we can thus
compute $\SPM_{\partial \P}$ in $O(\lambda_4(k)\log^2 k)$
time.

\subparagraph{Bounding the complexity of $\SPM_{\partial \P}$.} We now
prove that $|\SPM_{\partial \P}| = k = O(n^2h)$.

\frank{I think the following two lemmas could probably just move to a separate (earlier) subsection in which we prove the bound on the complexity of the spm.}
\maarten {I agree - I would find it more natural to first argue about the complexity, and only then talk about algorithms.}

\begin{lemma}
  \label{lem:function_appears_once_on_env1}
  For any three points $p_1,q,p_2$ in order on an edge $e$ and any signature $\sigma$ in $S_A(e,\delta)$, if $d_{\sigma|e}$ lies on $\env^{\leq\delta}(e,\delta)$ at both $p_1$ and $p_2$, but not at $q$, then there does not exist any signature $\tau\in S_A(e,\delta)$ such that $d_{\tau|e}$ lies on $\env^{\leq\delta}(e,\delta)$ at $q$.
  Symmetric properties apply to $\sigma\in S_B(e,\delta)$.
\end{lemma}

\begin{proof}
    We prove this by induction on the length of the signature $\sigma$.
    If $\sigma$ consists of just $s$, and $e$ is an edge incident to $s$, then $d_{\sigma|e}$ is simply the distance along $e$ from $s$.
    Assume without loss of generality that $q$ lies on the segment $\overline{sp_1}$ of $e$ (otherwise swap $p_1$ and $p_2$).
    If a signature $\tau$ as described would exist, then $d_{\tau|e}(p_2)\leq d_{\tau|e}(q)+\overline{qp_2}<d_{\sigma|e}(q)+\overline{qp_2}=d_{\sigma|e}(p_2)$, which would contradict that $d_{\sigma|e}$ lies on $\env^{\leq\delta}(e,\delta)$ at $p_2$.
    
    The case where $\sigma$ consists of just $s$, but $e$ is the edge of $T_s$ not incident to $s$ is an easier version\footnote{Because it forces $p'_1=q'=p'_2=s$.} of the case where $\sigma$ is not just $s$, which we hence omit.

    Now consider the case where $\sigma$ consists of not just $s$.
    Let $\pi_1$ and $\pi_2$ be shortest paths from $s$ to $p_1$ and $p_2$ whose signatures reduce to $\sigma$.
    Then $\pi_1$ has length $d_{\sigma|e}(p_1)$ and $\pi_2$ has length $d_{\sigma|e}(p_2)$.
    Let $e'$ be the last edge that these paths visit (before $e$), and let $p'_1$ and $p'_2$ respectively be their last point of intersection with $e'$.
    Now, assume for a contradiction that a signature $\tau$ as described does exist.
    Let $\rho$ be the corresponding path from $s$ to $q$ of length $d_{\tau|e}(q)$ whose signature reduces to $\tau$.
    By Lemma~\ref{lem:envIsSPM}, $\pi_1$, $\pi_2$, and $\rho$ are shortest paths.

    See Figure~\ref{fig:one_interval_env1} (a).\rodrigo{I guess the figure should adopt the $\self e$ and $\twin e$ notation?}
    If $\rho$ crosses $\overline{p'_1p_1}$ or $\overline{p'_2p_2}$ in a point $x$, then we can replace the part of $\pi_1$ or $\pi_2$ after $x$ to obtain a path from $s$ to $q$ of length at most $\|\rho\|$ whose signature reduces to $\sigma$, contradicting that $d_{\sigma|e}(q)>d_{\tau|e}(q)$.
    Because $\rho$ does not cross $\overline{p'_1p_1}$ or $\overline{p'_2p_2}$, $\rho$ must enter $T_A$ in a point $q'$ on $e'$ that lies between $p'_1$ and $p'_2$, see Figure~\ref{fig:one_interval_env1} (b).
    Let $T_C$ be the second triangle incident to $e'$.
    Let $\sigma'$ and $\tau'$ be signatures in $S_C(e',\delta)$ such that $d_{\sigma'|e'}(p'_1)$, $d_{\sigma'|e'}(p'_2)$, and $d_{\tau'|e'}(q')$ lie on $\env(e',\delta)$ (such signatures exist because the subpaths of $\pi_1$, $\pi_2$, and $\rho$ until $p'_1$, $p'_2$, and $q'$ are shortest paths of length at most $\delta$).
    Because $d_{\tau|e}(q)<d_{\sigma|e}(q)$, we have $d_{\tau'|e'}(q')<d_{\sigma'|e'}(q')$, and hence $d_{\sigma'|e'}$ does not lie on $\env^{\leq\delta}(e',\delta)$ at $q'$.
    By induction, there does not exist any signature $\tau''\in S_C(e',\delta)$ such that $d_{\tau''|e'}$ lies on $\env^{\leq\delta}(e',\delta)$ at $q'$, but $\tau'$ is such a signature, which is a contradiction.\qedhere
  \begin{figure}[tb]
    \centering
    \includegraphics{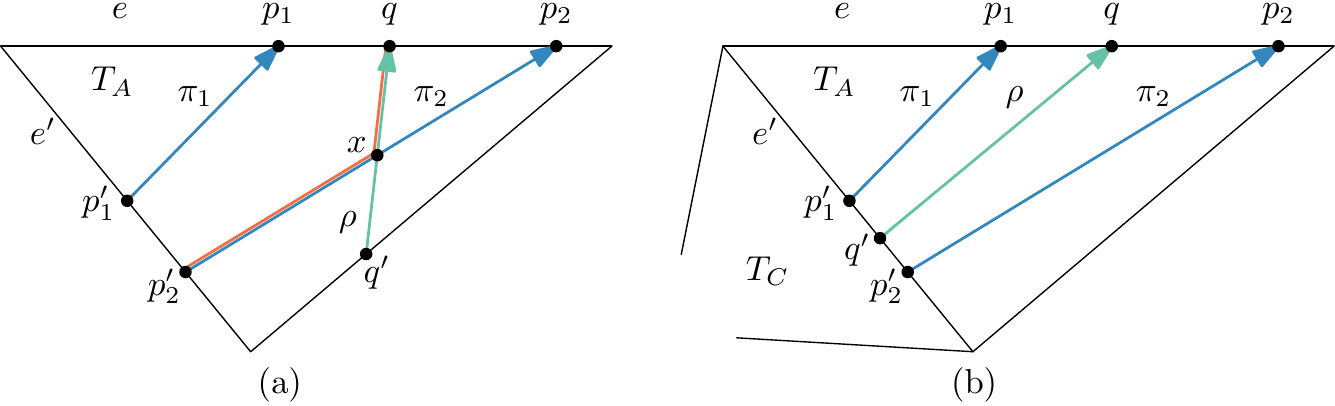}
    \caption{Cases of Lemma~\ref{lem:function_appears_once_on_env1}. Only the last segments of $\pi_1$, $\rho$, and $\pi_2$ are depicted.}
    \label{fig:one_interval_env1}
  \end{figure}
\end{proof}

\begin{lemma}\label{lem:linearEnv}
    The complexity of $\env^{\leq\delta}(e,\delta)$ is $O(|S_A(e,\delta)\cup S_B(e,\delta)|)$.
\end{lemma}
\begin{proof}
  Let $\Gamma$ be the sequence of signatures $\sigma$ in
  $S_A(e,\delta)\cup S_B(e,\delta)$ whose distance functions
  $f_{\sigma|e}$ appear on $\env^{\leq\delta}(e,\delta)$ as we
  traverse $e$ from left to right.  We argue that any contiguous
  subsequence $\Gamma'$ of $\Gamma$ of length $c+2$ contains a
  signature that did not appear on $\Gamma$ before, where $c = O(1)$
  is the number of times any pair of functions $f_{\sigma|e}$ and
  $f_{\sigma'|e}$ can intersect. Note that if a contiguous subsequence
  contains at least two signatures from $S_A(e,\delta)$ (or from
  $S_B(e,\delta)$), then by
  Lemma~\ref{lem:function_appears_once_on_env1}, the latter cannot
  have appeared on $\Gamma$ before.  On the other hand, if $\Gamma'$
  contains at most one signature from each $S_A(e,\delta)$ and
  $S_B(e,\delta)$, then the length of $\Gamma'$ is at most the
  complexity of the lower envelope of some $f_{\sigma|e}$ and
  $f_{\sigma'|e}$, which is at most $c+1=O(1)$.  It follows that the
  length of $\Gamma$, and hence the complexity of
  $\env^{\leq\delta}(e,\delta)$ is
  $O(|S_A(e,\delta)\cup S_B(e,\delta)|)$.
\end{proof}


\begin{lemma}
  \label{lem:spm_edge_complexity}
  For an $h$-happy portalgon, the complexity of the shortest path map restricted to the edges is $O(n^2h)$.
\end{lemma}

\begin{proof}
    The complexity of the shortest path map (restricted to edges) corresponds to the total complexity of $\envR(e,\infty)$ (summed over all edges $e$), which by Lemma~\ref{lem:linearEnv} is linear in the number of signatures appearing in $S_A(e,\infty)\cup S_B(e,\infty)$.
    For a signature $\sigma$ corresponding to a point on an edge $e$, the last element of $\sigma$ is~$e$, so points on distinct edges have distinct signatures.
    Hence, the complexity of the shortest path map is bounded by the total number of signatures.

    We bound the total number of signatures of shortest paths by considering the prefix tree of these signatures.
    Any edge appears at most $h$ times on the signature of any shortest path, so the length of any such signature is $O(nh)$.
    Hence, the depth of the prefix tree is $O(nh)$.

    Observe that any node in the prefix tree has $O(1)$ children: a prefix corresponds to a unique triangle, and can be extended only by appending one of the constantly many vertices or edges incident to that triangle.

    We show that only $O(n)$ nodes have multiple children.  Due to
    tie-breaking, any vertex $v$ corresponds to a unique shortest
    path, so the prefix tree contains a unique signature that ends in
    $v$.  We also show that for any triangle $T_A$ with edges $e$,
    $e_1$, and $e_2$, there is at most one prefix $\sigma$ ending in
    $e$ such that both $\sigma e_1$ and $\sigma e_2$ are prefixes of
    shortest paths.  Let $p_1$ and $p_2$ be points on $e$
    corresponding to shortest paths with signatures $\sigma e_1$ and
    $\sigma e_2$.  By Lemma~\ref{lem:function_appears_once_on_env1}
    the shortest paths entering $T_A$ between $p_1$ and $p_2$ (and
    thus all appear in $S_B(e,\infty)$, where $T_B\neq T_A$ is the
    other triangle incident to $e$) all have signature $\sigma$.

    Now suppose for a contradiction that there is a second signature $\sigma'$ such that both $\sigma' e_1$ and $\sigma' e_2$ are prefixes of shortest paths, then similarly, these paths enter $T_A$ through points $p'_1$ and $p'_2$, between which the signature is $\sigma'$.
    Therefore, the segments $\overline{p_1p_2}$ and $\overline{p'_1p'_2}$ are interior-disjoint, but this means that either the shortest paths with signatures $\sigma e_1$ and $\sigma' e_2$, or the shortest paths with signatures $\sigma e_2$ and $\sigma' e_1$ must cross, which is impossible.

    It follows that only $O(n)$ nodes of the prefix tree have multiple children.
    Because all nodes have constant degree, the prefix tree has only $O(n)$ leaves.
    Combining this with the $O(nh)$ depth, there are $O(n^2h)$ signatures in total over all shortest paths.
\end{proof}

%

We thus obtain the following result.

\begin{lemma}
  \label{lem:spm_edge_algo}
  Let \P be a triangulated $h$-happy portalgon with $n$ vertices and
  $m$ portals, and let $s$ be a given source point. The shortest path
  map $\SPM_{\partial \P}(s)$ of $s$ has complexity $k=O(n^2h)$ and
  can be computed in $O(\lambda_4(k)\log^2 k)$ time.
\end{lemma}

\subsection{Extension to the interior}
\label{app:Extending_to_the_Interior}

We now extend the shortest path map $\SPM_{\partial T}$ on the
boundary of a triangle $T$ into the shortest path map of the entire
triangle $T$. When $\SPM_{\partial T}$ consists of $m$ intervals,
$\SPM_T$ has complexity $O(m)$ and can be computed in $O(m\log m)$
time, for example using the \algcall{Build-Subdivision} algorithm of
Mitchell \etal~\cite[Lemma
8.2]{mitchell87discr_geodes_probl}. Applying this to all triangles in
the portalgon, we obtain the following result.


\begin{theorem}
  \label{thm:compute_SPM}
  Let \P be a triangulated $h$-happy portalgon with $n$ vertices and
  $m$ portals, and let $s$ be a given source point. The shortest path
  map $\SPM_\P(s)$ of $s$ has complexity $k=O(n^2h)$ and can be
  computed in $O(\lambda_4(k)\log^2 k)$ time.
\end{theorem}

\section{Existence of happy portalgons}
\label{sec:Existence_of_Happy_Portalgons}

In this section, we show that for every portalgon there exists an equivalent $O(1)$-happy portalgon; specifically, the {\em intrinsic Delaunay
  triangulation}~\cite{bobenko2007discrete} $\T$ of $\Sigma$ induces a happy
portalgon, whose fragments correspond to the triangles of $\T$.  The
resulting portalgon may have more fragments than the original, but its
total complexity is still linear.
While our proof is constructive, its running time may be unbounded.

  \subsection{Intrinsic Delaunay triangulation}
    
    Let $\T$ be a triangulated portalgon\footnote {We assume a triangulated portalgon for ease of argumentation in this section; note that any portalgon can easily be transformed into a triangulated portalgon by triangulating its fragments and replacing the resulting diagonals by portals.} 
    and let $\Sigma$ be its surface;
    let $V$ be the set of vertices of $\Sigma$.
    Intuitively, the intrinsic Delaunay triangulation of $\Sigma$ has a (straight) edge between two vertices $u, v \in V$ when there exists a circle with $u$ and $v$ on its boundary, and which contains no other vertices of $V$ when we ``unfold'' $\Sigma$, for example by identifying edges of $\T$ as in Section~\ref {sec:Shortest_Paths_in_Portalgons}.
    
    To simplify geometric arguments, we derive a simply-connected space $\hat\Sigma$ from $\Sigma$.
    We can think of $\hat\Sigma$ as the universal cover of $\Sigma\setminus V$, with vertices reinserted at the corresponding locations.
    More formally, we can define $\hat\Sigma$ by considering the directed fragment graph $G$ of $\T$. 
    \maarten [suggests] {Move the following notation to Section 2, where we define the fragment graph?}
    For any portal $e$ of $\T$, we denote by $T_{\self e}, T_{\twin e} \in \F$ the triangles (nodes of $G$) that contain the respective portal edge; $G$ has a link \maarten [suggests] {Should we use ``node'' and ``link'' for the fragment graph, to distinguish from ``vertex'' and ``edge'' of the portalgon?} $\dot e$ from $T_{\self e}$ to $T_{\twin e}$ and a link $\dot e^{-1}$ from $T_{\twin e}$ to $T_{\self e}$.
    We say that $\dot e$ and $\dot e^{-1}$ are inverses of each other.
    A {\em walk} in $G$ from a triangle $T$ to a triangle $T'$ is a (possibly empty if $T=T'$) sequence of links of $G$, such that $T$ is the source of the first link, and the source of the $(i+1)$-st is the target of the $i$-th link, and the target of the last link is $T'$.
    For a walk $w$ from $T$ to $T'$, we write $w\from T\to_G T'$ and say that a walk is \emph{backtracking} if it contains two consecutive links that are inverses of each other.
    Fix an arbitrary root triangle $T_0$, and for a walk $w\from T_0\to_G T$, let $T_w$ be a copy of the target triangle $T$ placed in the Euclidean plane by unfolding $T_0, T_1, \ldots, T_w$ along their common portals.
    \maarten {Changing notation from $\T$ to $\F$ for consistency, since if $\T$ is a portalgon, $\F$ is its set of triangles.}
    Let $\hat{\F}=\bigsqcup\{T_w\mid T\in \F, w\from T_0\to T, \text{$w$ is not backtracking}\}$ be the disjoint union of Euclidean triangles of target triangles of walks starting at $T_0$.
    We can now define $\hat\Sigma$ as $\hat{\F}/{\hat\sim}$, where for any two walks $w$ and $w'$, where $w'$ is obtained from $w$ by removing its final link, $\hat\sim$ glues $T_{w'}$ to $T_w$ along the sides corresponding to the last link of $w$.
    Let $q\from\hat\Sigma\to\Sigma$ be the map that sends points of $\hat\Sigma$ to their corresponding point in $\Sigma$.
    A map $\hat g\from X\to\hat\Sigma$ is a \emph{lift} of $g\from X\to\Sigma$ if $g=q\circ\hat g$.
    Any path $\pi\from[0,1]\to\Sigma$ has a lift $\hat\pi$ in $\hat\Sigma$.
    
    There is a map $f\from\hat\Sigma\to\Real^2$ whose restriction to any triangle of $\hat\Sigma$ is an isometry, and whose restriction to any pair of adjacent triangles is injective.
    We can think of $f$ as unfolding $\hat\Sigma$ so that it lies flat in the plane and is locally isometric everywhere except at the vertices (vertices are the only source of curvature), see Figure~\ref{fig-full:sigmahat}.

    \begin{figure}[b!]
        \centering
        $\vcenter{\hbox{\includegraphics[page=3]{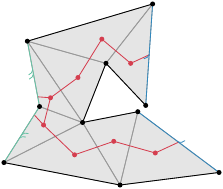}}}$
        \hfill
        $\vcenter{\hbox{\includegraphics[page=5]{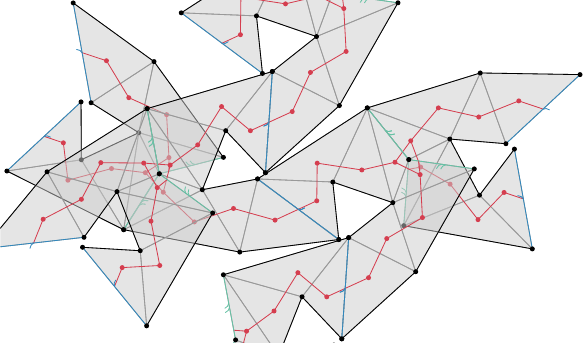}}}$
        \caption{$\Sigma$ (left) and a local region of $\hat\Sigma$ (right). The dual graph of $\hat\Sigma$ is an infinite tree (red).}
        \label{fig-full:sigmahat}
    \end{figure}
    
    \begin{observation}\label{obs:isometry}
       Let $B_r(x,y)=\{p\in\Real^2\mid\|p-(x,y)\|<r\}$ be the open disk in the plane of radius $r$ centered at~$(x,y)$.
       If a component $U$ of $f^{-1}(B_r(x,y))$ contains no vertices, then the restriction of $f$ to the closure of $U$ is an isometry.
    \end{observation}

    For a triangle $T$ of $\hat\Sigma$, let $D(T)$ be the open disk bounded by the circumcircle of $f(T)$, and let $C(T)$ be the closure of the component of $f^{-1}(D(T))$ that contains the interior of $T$.
    A triangle $T$ of $\hat\Sigma$ is \emph{Delaunay} if $C(T)$ does not contain any vertices of triangles adjacent to $T$ in its interior.
    $\T$ is an \emph{intrinsic Delaunay triangulation} of $\Sigma$  if and only if all triangles of $\hat{\F}$ are Delaunay.
    Lemma~\ref{lem:delaunay} generalizes a well-known property of Delaunay triangulations in the plane.
    \begin{lemma}[Bobenko \etal~\cite{bobenko2007discrete}]\label{lem:delaunay}
       If $\hat{\mathcal{T}}$ is an intrinsic Delaunay triangulation
       of $\hat\Sigma$, then for any $T \in \hat{\T}$, $C(T)$ contains no vertices in its interior.
    \end{lemma}
    Combining Observation~\ref{obs:isometry} and Lemma~\ref{lem:delaunay}, we obtain the following corollaries, see also Figure~\ref{fig-full:emptycircle}.
    \begin{corollary}\label{cor:circum}
       For any triangle $T$ of an intrinsic Delaunay triangulation of $\hat\Sigma$, the restriction of $f$ to $C(T)$ is injective, and shortest paths intersect $C(T)$ in straight segments.
    \end{corollary}
    \begin{corollary}\label{cor:boundaryChords}
       For any triangle $T$ of an intrinsic Delaunay triangulation of $\hat\Sigma$, $\partial f(C(T))$ is a union of chords and circular arcs of $\partial D(T)$, where each chord is a boundary edge.
    \end{corollary}
    \begin{figure}
        \centering
        \includegraphics[page=6]{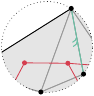}
        \caption{Although $f$ is not injective in general, for a triangle $T$ of a Delaunay triangulation, the restriction of $f$ to $C(T)$ is injective and contains no vertices of the triangulation in the interior.}
        \label{fig-full:emptycircle}
    \end{figure}

\subsection{Intrinsic Delaunay triangulations are happy}

    We are now ready to prove the main result of this section: the intrinsic Delaunay triangulation of any portalgon has constant happiness. \maarten {Do we know the constant? Or at least an upper bound?}

    Let $\mathcal{T}$ be an intrinsic Delaunay triangulation of $\Sigma$.
    We want to bound the number of intersections between a shortest path $\pi$ on $\Sigma$ and edges of the triangulation.
    For this, let $\pi$ be a shortest path between two given points on $\Sigma$, and among all such paths, assume that $\pi$ has a minimum number of crossings with edges of $\Sigma$.
    Here, we count crossings with an edge as the components of intersection with that edge.

    
    Now consider an arbitrary edge $e$ between two vertices $u$ and $v$ of the triangulation and let $m$ be its midpoint.
    We will show that $\pi$ intersects $e$ only constantly often.
    For a contradiction, suppose that $\pi$ intersects $e$ at least $7$ times.
    If a component of intersection of $\pi$ with $e$ is not transversal and does not contain a vertex, then it follows from Corollary~\ref{cor:circum} that $\pi$ is a segment of $e$, so $\pi$ has only one edge crossing.
    If a non-transversal intersection contains both vertices of $e$, then $\pi$ contains $e$ and hence intersects $e$ only once.
    Thus, every non-transversal intersection of $\pi$ with $e$ contains exactly one vertex of $e$, and hence contains the start or end of $\pi$.
    Therefore, $\pi$ intersects $e$ non-transversally at most twice.
    
    To bound the total number of intersections with $e$ we analyze the geometry of $\hat\pi$ in a local neighborhood of a lift $\hat e$ of $e$.
    Arbitrarily fix one of the triangles $T$ incident to $e$ and consider the neighborhood $C(\hat T)$ of $\hat e$ for the corresponding lift $\hat T$ incident to $\hat e$.
    By Corollary~\ref{cor:circum}, the restriction of $f$ to $C(\hat T)$ is an isometry, so $f(\Img(\hat\pi)\cap C(\hat T))$ is a straight line segment.
    Let $p_i$ be the $i$-th point of (or vertex of non-transversal) intersection of $\pi$ with $e$, and define $\lambda_i$ such that $\pi(\lambda_i)=p_i$.
    Let $\hat e_i$ be a lift of $e$ containing $\hat\pi(\lambda_i)$, and let $\hat p_i$, $\hat T_i$, $\hat u_i$, $\hat m_i$, $\hat v_i$ be the respective lifts of $p_i$, $T$, $u$, $m$, $v$ incident to $\hat e_i$, see Figure~\ref{fig-full:inward} (left).
    Let $D_i:=D(\hat T_i)$, and $f_i$ be the restriction of $f$ to $C_i:=C(\hat T_i)$.
    Let $c_i$ be the center of $D_i$, i.e., the circumcenter of $f(\hat T_i)$.
    \begin{figure}
        \centering
        \includegraphics[page=7]{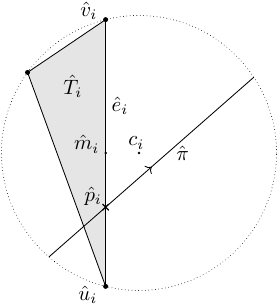}~
        \includegraphics[page=8]{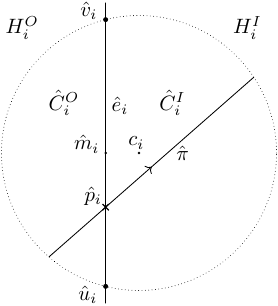}
        \caption{The crossing $p_i$ of $\pi$ with $e$ is inward.}
        \label{fig-full:inward}
    \end{figure}
    
    Define $H_i^I$ and $H_i^O$ to be the two half-planes bounded by the line through $f(\hat e_i)$, such that $H_i^I$ contains $c_i$ (if $c_i=f(\hat m_i)$, label the half-planes by $H_i^I$ and $H_i^O$ arbitrarily).
    We respectively call $H_i^I$ and $H_i^O$ the \emph{inner} and \emph{outer} half-plane of $\hat e_i$, and define $C_i^I:=f_i^{-1}(H_i^I)$ to be the \emph{inner} component of $\hat e_i$, and $C_i^O:=f_i^{-1}(H_i^O)$ to be the \emph{outer} component of $\hat e_i$.
    \frank{some more intuitive descriptions of what these
      components/sets $C_i^I$ and $C_i^O$ look like/represent would be
    nice. Currently, they are hard to describe/envision.}

    We call the crossing $p_i$ \emph{inward} if $p_i$ is a non-transversal crossing or $\hat\pi$ crosses $\hat p_i$ from $C_i^O$ to $C_i^I$, see Figure~\ref{fig-full:inward} (right).
    
    Let $s$ be the segment of $e$ from $u$ to $m$.
    Assume without loss of generality that at least four of the (at least seven) crossings of $\pi$ with $e$ lie on $s$ (otherwise relabel $u$ and $v$).
    \begin{lemma}\label{lem:inward}
       For any inward crossing $p_i$ of $\pi$ with $s$, none of the crossings $p_j$ with $j>i$ lie on the segment of $s$ between $u$ and $p_i$.
    \end{lemma}
    \begin{proof}
        Suppose for a contradiction that for some $j>i$, $p_j$ lies on the segment of $e$ between $u$ and $p_i$.
        Let $Z$ be the open disk centered at $f(\hat m_i)$ with $f(\hat e_i)$ as diameter.
        \frank{drawing $Z$ in some figure would be helpful.}
        The path $\hat m_i\xrsquigarrow{\hat e_i}\hat p_i\xrsquigarrow{\hat\pi}\hat p_j\xrsquigarrow{\hat e_j}\hat u_j$ is strictly shorter than $s$, so the segment $f(\hat p_j\xrsquigarrow{\hat e_j}\hat u_j)$ lies interior to~$Z$.
        Note that $H_i^I\cap Z$ lies completely inside $D_i$, and hence does not contain any vertices of $C_i^I$ in its interior, so in particular $\hat u_j$ does not lie in $C_i^I$.
        Moreover, any path that enters the interior of $C_i^I$ cannot leave it without leaving $Z$ or crossing the interior of $\hat e_i$ (because $\hat e_i$ is the only common boundary of $C_i^I$ and $\hat\Sigma\setminus C_i^I$ that lies inside $Z$).
        However, the path $\hat p_i\xrsquigarrow{\hat\pi}\hat p_j\xrsquigarrow{\hat e_j}\hat u_j$ enters and leaves $C_i^I$ without leaving $Z$.
        The path $\hat p_i\xrsquigarrow{\hat\pi}\hat p_j$ does not repeatedly cross $\hat e_i$, so the segment of $\hat e_j$ between $\hat p_j$ and $\hat u_j$ must cross $\hat e_i$, so $\hat e_j$ crosses the interior of $\hat e_i$, but this is a contradiction because edges of the triangulation are interior-disjoint.
    \end{proof}
    If $p_i$ of $\pi$ is not inward, then it is inward on the reverse of $\pi$.
    Corollary~\ref{cor:notInward} follows.
    
   \begin{corollary}\label{cor:notInward}
       If some crossing $p_i$ of $\pi$ with $s$ is not inward, then none of the crossings $p_h$ with $h<i$ lie on the segment of $s$ between $u$ and $p_i$.
   \end{corollary}
    \begin{figure}[b]
        \centering
        \includegraphics[page=13]{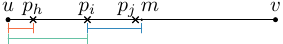}
        \caption{A subsequence of crossings $p_h$, $p_i$, $p_j$ of $\pi$ with $s$ whose distance from $u$ is increasing.}
        \label{fig-full:increasingSequence}
    \end{figure}
    Lemma~\ref{lem:inward} implies that for the sequence of intersections $p_i$ of $\pi$ with $s$, the distance function from $u$ to $p_i$ along $s$ has only one local minimum.
    So there exists a subsequence of at least three crossings $p_h$, $p_i$ and $p_j$ ($h<i<j$) of $\pi$ with $s$ such that $p_i$ lies between $p_h$ and $p_j$.
    Assume without loss of generality that the distances (along $s$) from $u$ to $p_h$, $p_i$, and $p_j$ are increasing (the other case follows by considering the reverse of $\pi$). 
    At this point, we observe:
    \begin {observation}
     $f(\hat m_i) \ne c_i$
    \end {observation}
    \begin {proof}
      Suppose for contradiction that $f(\hat m_i)$ is $c_i$. 
      In that case we labeled $H_i^I$ and $H_i^O$ arbitrarily, and the other labeling tells us that the distances from $u$ to $p_h$, $p_i$, and $p_j$ must be decreasing instead of increasing, which is impossible if they were increasing to begin with (switching the labeling does not change the distances from $u$).
    \end {proof}
    
    By Corollary~\ref{cor:notInward}, $p_i$ is an inward crossing.
    We show how the possible locations of the subsequent crossings inside $C_i$ are constrained.
    Define $D'_i$ to be the disk concentric with $D_i$, whose boundary passes through $f(\hat m_i)$.
    Any chord of $D_i$ that passes through $D'_i$ is at least as long as $e$.
    Let $\overline e_h:=f(\hat e_h)\cap D_i$. We observe:
    \begin {lemma}
     $\overline e_h$ cannot intersect the closure of $D'_i$.
    \end {lemma}
    
    \begin {proof}
    Let $R_i$ (shaded in Figure~\ref{fig-full:proof} (left)) be the interior of the convex hull of $D'_i\cup f(\hat e_i)$.
    Because $C_i$ has no interior vertices, no lift of $e$ intersects $f_i^{-1}(R_i)$, see Figure~\ref{fig-full:proof}~(left).
    In correspondence with the figures, we will without loss of generality assume that $f(\hat e_i)$ is vertical, with $f(\hat u_i)$ on the bottom, and $H_i^O$ on the left.
    \begin{figure}[h]
        \centering
        \hspace{.85em}
        \includegraphics[page=13]{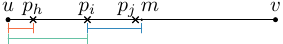}\\[.5em]
        \hspace{-1em}
        \includegraphics[page=9]{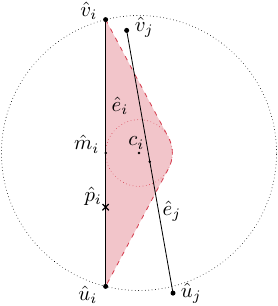}~
        \includegraphics[page=10]{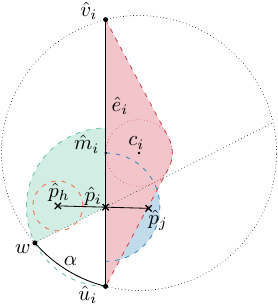}~
        \includegraphics[page=11]{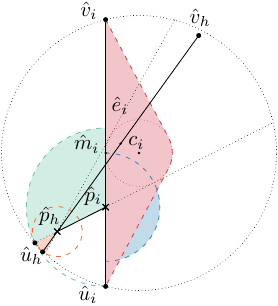}
        \caption{
            (Left) There is no lift of $e$ that intersects the interior of the shaded region $R_i$ of $C_i$.
            (Middle) $\hat p_j$ lies in the blue region, $\hat p_h$ in the green region.
            (Right) $\hat u_h$ cannot lie interior to $C_i$.}
        \label{fig-full:proof}
    \end{figure}
    
    The fact that the subsequent intersection $p_j$ of $\pi$ with $e$ lies on the segment of $e$ between $p_i$ and $m$ constrains the angle at which $\pi$ can intersect $p_i$.
    Indeed, $f\circ\pi$ must leave $R_i$ before the next intersection, and the distance from $\hat p_i$ to $\hat p_j$ is less than that from $f(\hat p_i)$ to $f(\hat m_i)$.
    Figure~\ref{fig-full:proof} (middle) illustrates the possible locations of $\hat p_j$.
    Crucially, the chord (call it $\overline\pi$) of $D_i$ that contains $f(\hat\pi\cap C_i)$ does not intersect $D'_i$.
    
    Because $f(\hat m_i)\neq c_i$, two chords of $D_i$ pass through $f(\hat p_i)$ and are tangent to $D'_i$.
    One of those chords is $f(\hat e_i)$, and the other chord (call it $\omega$) has one endpoint (call it $w$) interior to $H_i^O$.
    Because $\overline\pi$ does not intersect $D'_i$, the point $f(\hat p_h)$ lies on or above the line through $w$ and $f(\hat p_i)$.
    Then $\overline e_h$ is a chord of $D_i$, and it cannot intersect the closure of $D'_i$.
    \end {proof}
    We will arrive at a contradiction to our initial assumption  by showing:
    \begin{lemma}
      \label{lem:intersects}
      If $\pi$ crosses $e$ at least $7$ times then
      $\overline e_h$ must intersect the closure of $D'_i$.
    \end{lemma}
    
    \begin {proof}
    Let $q\in\partial D_i$ be the point of $\overline e_h$ closest to $f(\hat u_h)$.
    We will analyze the possible locations of $q$.
    Let $d_i$ be the distance from $f(\hat p_i)$ to $f(\hat u_i)$.
    Let $B_i$ be the disk of radius $d_i$ centered at $f(\hat p_i)$, so that $B_i$ intersects $\partial D_i\cap H_i^O$ in the arc (labeled $\alpha$ in Figure~\ref{fig-full:proof}) connecting $w$ and $f(\hat u_i)$.
    Let $b$ be the point inside $H_i^O$ where $\overline\pi$ intersects $\partial B_i$.
    Let $B_h$ be the disk centered at $f(\hat p_h)$ whose boundary intersects $b$, so that $B_h\subseteq B_i$.
    Let $d_s(\cdot,\cdot)$ be the distance along $s$, then $q$ lies in $B_h$ because
    \[
        d(q,f(\hat p_h))\leq
        d_s(u,p_h)=
        d_s(u,p_i)-d_s(p_i,p_h)\leq
        d_s(u,p_i)-d(f(\hat p_h),f(\hat p_i))=
        d(b,f(\hat p_h)).
    \]
    Thus, $q$ lies on the arc (call it $\gamma$) in which $B_h$ intersects $\partial D_i\cap H_i^O$.
    This arc lies on the arc between $f(\hat u_i)$ and $w$.
    Let $\gamma_0$ be the point of $\gamma$ closest to $w$, $\gamma_1$ be the point closest to $f(\hat u_i)$, and $\gamma_{\frac 12}$ be the midpoint of $\gamma$, see Figure~\ref{fig-full:proof_pt2}.
    The line through $\gamma_{\frac 12}$ and $f(\hat p_h)$ passes through the center $c_i$ of $D'_i$, so the chord $\overline e_h$ intersects the vertical line through $c_i$ at or below $c_i$.
    \begin{figure}[h]
        \centering
        \includegraphics{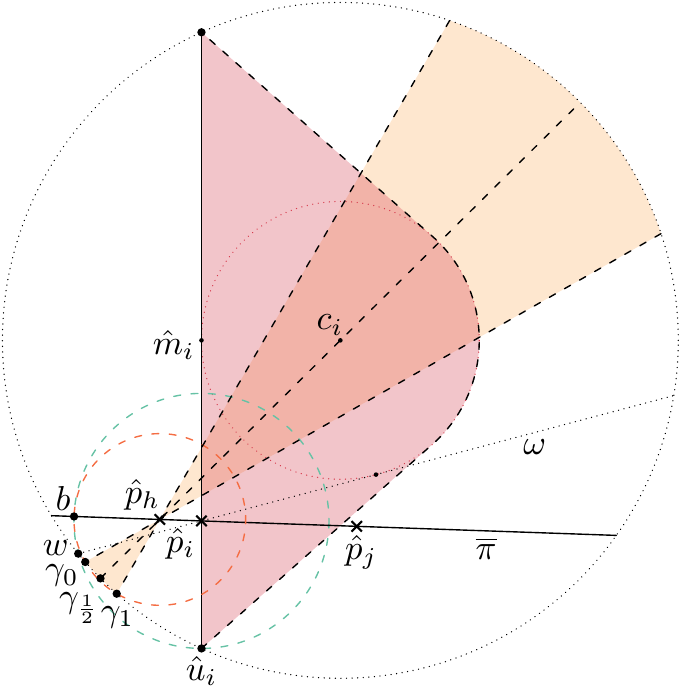}
        \caption{Any chord of $D_i$ through $f(\hat p_h)$ and $\gamma$ intersects $D'_i$.}
        \label{fig-full:proof_pt2}
    \end{figure}

    We distinguish two cases, first consider the case where $q$ lies on the arc of $\gamma$ between $\gamma_0$ and $\gamma_{\frac 12}$.
    Note that $\overline e_h$ is the chord of $D_i$ that passes through $q$ and $f(\hat p_h)$.
    Because $f(\hat p_h)$ lies on or above $\omega$, the segment of $\overline e_h$ from $q$ to $f(\hat p_h)$ either coincides with $\omega$ or intersects $\omega$ from below, so the remainder of $\overline e_h$ coincides with $\omega$, or lies above $\omega$.
    In particular, $\overline e_h$ lies above the point where $\omega$ is tangent to $D'_i$.
    However, any such line that also passes on or below $c_i$ must intersect $D'_i$, so $\overline e_h$ intersects $D'_i$ if $q$ lies on the arc of $\gamma$ between $\gamma_0$ and $\gamma_{\frac 12}$.
    
    The case where $q$ lies on the arc between $\gamma_{\frac 12}$ and $\gamma_1$ follows by the symmetry of $D'_i$, $f(\hat p_h)$, and $\gamma$ after reflection through the line through $f(\hat p_h)$ and $c_i$.
    Hence, the chord $\overline e_h$ intersects $D'_i$.
  \end {proof}  
    By Lemma~\ref {lem:intersects} the chord $\overline e_h$ intersects $D'_i$, which is a contradiction.
    We conclude that $\pi$ intersects $e$ at most six times.
    %
    We summarize our discussion in the following theorem. 
    
\begin{theorem}
  \label{thm:universal_happiness}
  Let \P be a portalgon with $n$ vertices.
  There exists a portalgon $\P' \equiv \P$ with $O(n)$ vertices that is $O(1)$-happy.
\end{theorem}

\section{Making portalgons happy}
\label{sec:Making_Portalgons_Happy}

In this section we study conditions for a fragment to be happy, and we
present a method to rearrange an unhappy fragment into an equivalent
one that is happy.

We start observing that the approach from Section~\ref{sec:Existence_of_Happy_Portalgons} is, in principle, constructive.
However, its running time depends on the number of edge flips required to reconfigure an initial triangulation of the fragments into the intrinsic Delaunay triangulation, and this number may not be expressible in terms of the input complexity. Whether the intrinsic Delaunay triangulation can be computed in a different way is an open question.
We note that when there is a bound on the minimum angle on the input triangles, the number of flips can be bounded~\cite {anglebounds}.

The rest of this section is devoted to an analysis of what we can do when there is no such minimum angle bound.
The main result  in this section is that we can reduce any portalgon for which the fragment graph contains at most one cycle to the case of a portalgon with just one fragment and a single portal, which we can rearrange to have only fragments of constant happiness (Theorem~\ref{thm:makemehappy}).

We begin with a simple observation that will be useful.

\begin{observation}
  \label{obs:single_portaledge_fragments}
  Any fragment $F$ with at most one portal edge $e$ is $2$-happy.
\end{observation}

\begin{proof}
  Assume by contradiction that $F$ is not 2-happy. Then
  there must be two points $p,q \in \P$ for which $\geod(p,q)$ crosses
  $e$ at least three times, and thus there is a subpath $\geod(r,s)$
  of $\geod(p,q)$ such that: $r$ and $s$ lie on $e$,
  $\geod(r,s) \subseteq F$, and $\geod(r,s)$ contains a point in the
  interior of $F$. By triangle inequality $\overline{rs}$ is then
  shorter than $\geod(r,s)$. Contradiction.
\end{proof}

  \label {sec:par}
  We now consider the situation in which there is only a single portal, and give a constructive result on how to compute an equivalent happy portalgon.
  First, we focus on the easier case where the two portal edges are parallel.
  In Appendix~\ref {app:nonpar} we extend this result to non-parallel edges.
  We also note there that if the angle between the two edges is at least a constant, then the fragment is already happy.
  Thus the near-parallel case is of most interest.
  
  \subsection{Single portal, parallel edges}
\subsubsection{Analysis}

  The happiness of a fragment with a single portal depends on the {\em shift} of the two
  (parallel) portal edges. The shift of two portal edges $\self e, \twin e$ can be
  defined as the distance between the two perpendiculars to the portal
  edges that go through the start vertices of $\self e$ and $\twin e$,
  respectively. In the following we assume without loss of generality
  that the portals are horizontal, thus the shift $\Delta$ is simply
  the difference in $x$-coordinate between the start vertices of the
  portal edges.
  Further, let $v$ denote the vertical distance between 
    $\self e$ and $\twin e$.

 \begin {lemma}
   \label {lem:h0}
   Let $F$ be a fragment with exactly two portal edges $\self e, \twin e$, which are parallel and belong to the same portal. If the shift $\Delta$ of $\self e, \twin e$ is
   $0$, then the fragment is $2$-happy.
 \end {lemma}
    
    \begin {proof}
      Let $F$ be a fragment with two parallel portal edges
      $\self e, \twin e$ with a shift of $\Delta=0$, and recall that $\self e$
      and $\twin e$ are assumed to be horizontal.
      
      Suppose that there are two points $p, q \in F$ whose shortest
      path crosses $\self e$ (and thus also $\twin e$) twice, say first at
      point $r$ and then at $s$.
      Consider the subpath of the shortest path  from
      $\twin r$ to $s$.  This subpath must have length at least
      $|\twin rs|$.  Since $e$ and $e'$ are parallel and $\Delta=0$,
      $\triangle r \twin r s$ has a right angle at $r$ and
      $\twin rs$ is its hypotenuse, hence $|\twin rs| > |\self rs|$. Thus
      going from $r$ to $s$ along $e$ is shorter, reaching a
      contradiction.  It follows that $F$ is  $2$-happy.
    \end {proof}

    If $\Delta \ne 0$, a fragment might be happy or not, depending on several other circumstances: the length of the portals, the distance between the portal edges, and whether the boundary of $F$ interferes or not. 
    In the following we present a method to transform a fragment that is not happy into an equivalent portalgon that is happy.

    The general idea is to create a new portal by cutting $F$ through
    a line in the direction orthogonal to the line through the two
    start vertices of the portals.  
    
    First we present the idea for the case where $F$ is a parallelogram;
    we refer to its two parallel portal edges by $\self e$ and $\twin e$.
    
    Assume without loss of generality
    that $\self e$ is above $\twin e$, and that $\self e$ starts to the left of
    $\twin e$.  In this case, the slope of the line in the direction
    orthogonal to the portal start points is $z = -\Delta/v$.  
    We begin at the leftmost vertex of $\twin e$, and shoot a ray with
    slope $z$ in the inside of $F$ until we hit the boundary. 
    Every time the ray crosses the portal, we ``cut'' along this ray, creating a new portal along it.
    This results in several smaller fragments,
    which we then glue together again along the pieces of the original portals, into a fragment.
    The resulting fragment is a
    rectangle $F'$. 
    See Figure~\ref{fig-full:example_glue} for an
    illustration. Note that, by definition of $z$, this new fragment $F'$ now has shift zero. Hence:

   \begin {lemma}
     \label {lem:parallelogram}
     For any parallelogram with two parallel portal edges, there
     is an equivalent parallelogram that has $\Delta=0$ and  therefore is $2$-happy.
   \end {lemma}

\begin{figure}[tb]
  \centering
     \includegraphics {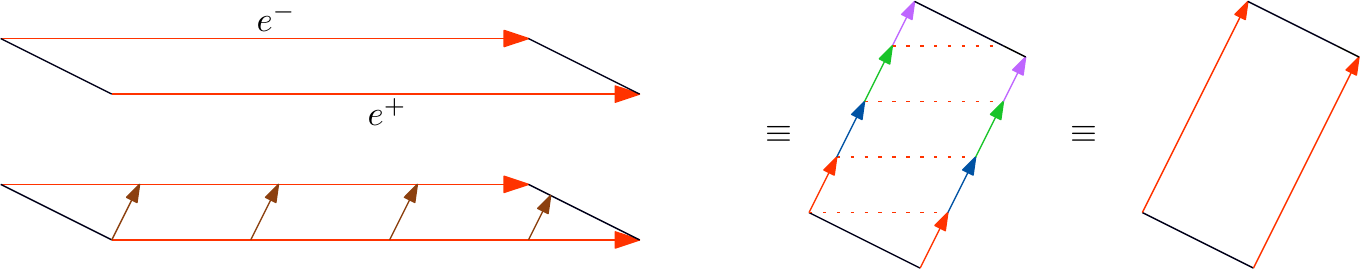}
  \caption{Left: a fragment with two parallel edges and non-zero shift. Center: result of cutting the fragment along a perpendicular ray, resulting in a new fragment with several portals, which is equivalent to a fragment with one portal and zero shift (right).}
  \label{fig-full:example_glue}
\end{figure}

Note that for a fragment with two parallel portals and non-zero shift, there is a unique 
equivalent fragment with zero shift, which is the one obtained by cutting along the perpendicular ray.
However, if $F$ is not a parallelogram, it may occur that when
gluing together the new smaller fragments we obtain a non-simple
polygon. Therefore we may not be able to transform $F$ into a single
equivalent $h$-happy fragment, as shown in the example in Figure~\ref{fig-full:bulge}. In that case, cutting along the perpendicular ray produces a non-simple fragment.

 \begin{figure}[tb]
  \centering
      \includegraphics {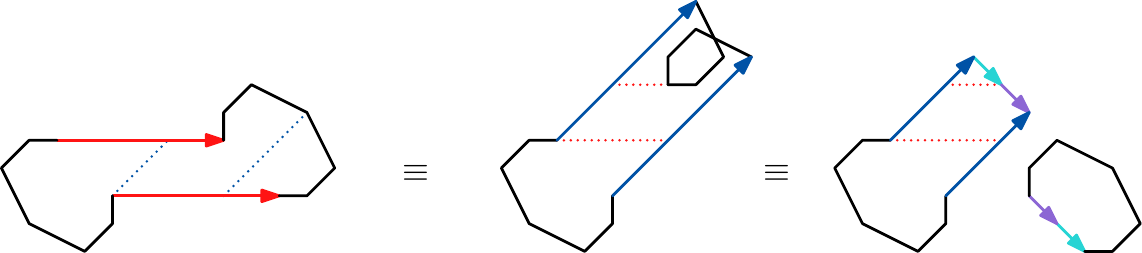}
      \caption{Example of a fragment (left) where the technique of cutting along the perpendicular ray produces a non-simple polygon (center). However, this can be fixed by using one more fragment (right).}
      \label{fig-full:bulge}
\end{figure}

\begin {observation}
  There exists a fragment with two horizontal portal edges for which
  there is no equivalent single fragment with $\Delta=0$.
\end {observation} 
%

Fortunately, we can transform any fragment with two parallel portal edges into a constant
number of $h$-happy fragments.
    First, we make a useful observation.

    \begin {lemma}
      \label {lem:cat}
      Consider a fragment $F$ with two horizontal portals $\self e$
      and $\twin e$.
      Let $m$ be the line through the start points of $\self e$ and       $\twin e$.
      %
      %
        If there is no line
      segment parallel to $m$ from a point on $\self e$ to a point on
      $\twin e$ whose interior lies in the interior of $F$, then $F$ is
      $5$-happy.
    \end {lemma}
    
    \begin{proof}
      \begin{figure}[tb]
        \centering
        \includegraphics{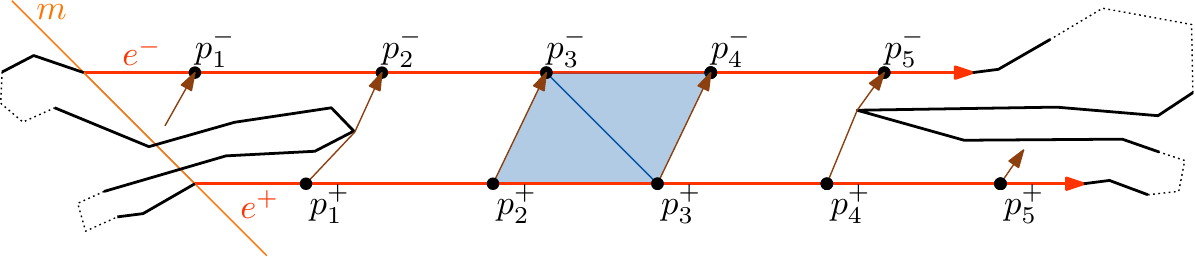}
        \caption{Illustration of the proof of Lemma~\ref{lem:cat}; $\pi$ is shown in brown. The
          segment $\self {p_3}\twin {p_3}$ is parallel to $m$ and is fully contained in the fragment.}
        \label{fig-full:parallelogram_m}
      \end{figure}

      Suppose that $F$ is not $5$-happy. 
      Then there is a shortest
      path $\geod$ whose intersection with $F$ consists of at least
      six connected components, and thus $\pi$ crosses the portal at
      least five times, in points
      $p_1,\ldots,p_5$. See
      Figure~\ref{fig-full:parallelogram_m}. Since $\geod$ is a shortest path, these components cannot share vertices (otherwise we could
      simply shorten the path). 
      It follows that $\geod(\twin {p_2},\self {p_3})$ must be a line segment (otherwise, if it would include a vertex from $F$, such a vertex would be shared with $\geod(\twin {p_1},\self {p_2})$ or with $\geod(\twin {p_3},\self {p_4})$).
      Analogously, $\geod(\twin {p_3},\self {p_4})$ must be a line segment.
      Furthermore, $\geod(\twin {p_2},\self {p_3})$ and $\geod(\twin {p_3},\self {p_4})$ are parallel, and thus $\twin {p_2} \self {p_3} \self {p_4} \twin {p_3}$
      is an empty parallelogram. Observe that, since $\self {p_3}$ and $\twin {p_3}$ are corresponding points (i.e., twins), 
      the segment between them has the same slope as $m$.
      Therefore, there is a line segment parallel to $m$ whose interior is inside $F$.
    \end{proof}

    If there is a segment parallel to $m$ contained in $F$ connecting $\self e$ to $\twin e$, then in particular there is also a leftmost and a rightmost such segment (we consider here $F$ as a closed set). 
    These segments, together with the pieces
    of portal edges, form an empty parallelogram. 
    Let $Z$ be this
    parallelogram, bounded by two horizontal edges and two edges parallel to $m$.

    \begin{observation}
      \label{obs:shortest_path_component}
      Let $F$ be a simple fragment, let $\overline{ab}$ be a line
      segment inside $F$. A single connected component $\geod \cap F$
      of a shortest path $\geod$ can intersect $\overline{ab}$ in at
      most one segment.
    \end{observation}

    In particular, Observation~\ref{obs:shortest_path_component}
    implies that the two endpoints $p$ and $q$ of a maximal component
    of $\geod \cap F$ cannot lie on the same portal edge $e$ unless
    $\geod(p,q) = \overline{pq}$.

   \begin {lemma}
     \label {lem:happy_portalgon_exists_parallel}
     Let \P be a portalgon with one fragment $F$ with $n$
     vertices, and one portal whose edges are parallel. There exists
     a $5$-happy portalgon $\P'$ equivalent to \P consisting of at
     most three fragments and total complexity $O(n)$.
   \end {lemma}

    \begin {proof}
      Assume without loss of generality that both portal edges $\self e$ and
      $\twin e$ are horizontal and oriented left-to-right. We now argue
      that when no three vertices of $F$ are colinear, and $F$ is not
      already $5$-happy, we can split $F$ into at most seven $4$-happy
      fragments of total complexity $O(n)$. Refer to
      Figure~\ref{fig-full:socks}. Finally, we show how to reduce the
      number of fragments to three, while remaining $4$-happy, even
      without the general position assumption.

      Let $m$ be the line through the start points of $\self e$ and
      $\twin e$.  By Lemma~\ref{lem:cat}, if there is no translate of
      $m$ whose intersection with $F$ contains a segment connecting
      $\self e$ to $\twin e$, $F$ is already $5$-happy. Let $m_\ell$ be the
      leftmost such translate of $m$ and $m_r$ the rightmost such
      translate; $m_\ell$ contains a vertex $\ell$ of $F$ and $m_r$
      contains a vertex $r$ of $F$ (possibly, $\ell$ or $m$ is an
      endpoint of $\self e$ or $\twin e$).  Let $\self a$ and $\twin a$ be the
      intersection points of $m_\ell$ with $\self e$ and $\twin e$, and let
      $\self b$ and $\twin b$ be the intersection points of $m_r$ with $\self e$
      and $\twin e$.  We cut the parallelogram
      $Z = \self a \self b \twin a \twin b$ from $F$, which splits $F$ into at most
      seven fragments (since, by general position,
      $\overline{\self a\twin a}$ and $\overline{\self b\twin b}$ contain at
      most two reflex vertices each). We now transform $Z$ into a
      $2$-happy fragment using Lemma~\ref{lem:parallelogram}.
      %
      Let $T$ and $B$ be the fragments containing the starting points
      of $\self e$ and $\twin e$, respectively. We argue that $T$ is
      $4$-happy. The argument that $B$ is $4$-happy is symmetric. The
      same holds for the fragments containing the endpoints of $\self e$ and
      $\twin e$. Any other fragments (if they exist) contain only one
      portal edge and are thus $2$-happy by
      Observation~\ref{obs:single_portaledge_fragments}.

  \begin{figure}[tb]
\centering
\includegraphics{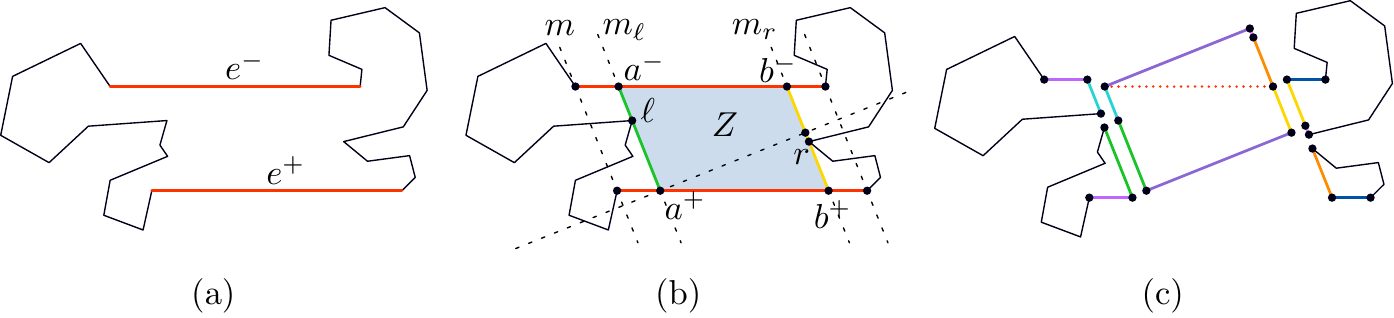}

\caption
{
        (a) Fragment $F$ with two parallel portals.
        (b) The lines $m_\ell$ and $m_r$ intersecting $\self e$ and $\twin e$ in
        $\self a,\twin a$ define a parallelogram $Z= \self a \self b \twin a \twin b$ that
        splits $F$ into at most seven sub fragments.
        (c) The resulting set of $5$-happy fragments. Note that it is
        possible to reduce the number of fragments by shifting $m_\ell$
        slightly to the right and $m_r$ slightly to the left.
      }
\label{fig-full:socks}
\end{figure}

      Consider the maximal connected components of a shortest path
      $\geod = \geod(s,t)$ with $F$. By
      Observation~\ref{obs:shortest_path_component} such a component
      either: (i) contains $s$, (ii) contains $t$, or (iii) connects a
      point $p_i$ on $\self e$ to a point $q_i$ on $\twin e$. Again by
      Observation~\ref{obs:shortest_path_component} each such
      component can intersect $\overline{\self a\twin a}$ at most once, so
      each such component can intersect $T$ at most once.

      We now further classify the type (iii) components into three
      types, depending on whether $p_i$ lies on the part of $\self e$ in $T$
      and whether $q_i$ lies on the part of $\twin e$ in $B$.

      If $p_i$ lies outside $T$, then
      Observation~\ref{obs:shortest_path_component} implies that
      $\geod(p_i,q_i)$ does not intersect (the interior of) $T$ at all
      since $\geod(p_i,q_i)$ would have to intersect
      $\overline{a\ell}$ twice.

      If $p_i$ lies in $T$ and $q_i$ lies in $B$, then
      Observation~\ref{obs:shortest_path_component} implies that
      $\geod(p_i,q_i)$ contains point $\ell$. Hence, there can be at
      most one such component of this type that intersects $T$.

      If $p_i$ lies in $T$ but $q_i$ lies outside of $B$, we have that
      $q_i = \twin t$ for some point $\self t$ on $e$. 
      If the component
      starting in $\self t$ is again of type (iii) (i.e., $\self t=p_{i+1}$) it
      cannot intersect $B$ as this would imply that
      $\geod(p_{i+1},q_{i+1})$ intersects $\geod(p_i,q_i)$. Nor can it
      intersect $T$ (since $q_{i+1}$ lies on $\twin e$ it would have
      to intersect $\overline{a\ell}$ twice). The same argument holds
      for any component on $\geod(p_j,q_j)$ with $j > i$. Hence,
      $\geod(p_i,q_i)$ intersects $T$ only if this is the last
      component of type (iii). Clearly, there is only one such a
      component.

      It follows that there are only four components of
      $\geod \cap F$ that intersect $T$, and each such component
      intersects $T$ in only one consecutive subpath. Hence $T$ is
      $4$-happy.

      We conclude that the portalgon \P' that we obtain is
      $4$-happy, equivalent to \P, and has at most seven
      fragments (note that there will be seven fragments if there are two reflex vertices on $m_\ell$, and two on $m_r$). 
      Furthermore, every vertex of \P appears in at most
      $O(1)$ fragments of $\P'$, and thus $\P'$ has complexity $O(n)$.

      Finally, observe that (before splitting $F$) we can actually
      shift the left and right sides of $Z$ inwards by some
      arbitrarily small $\eps$ (in particular, something smaller than
      the smallest distance between two non-adjacent edges of $Z$). It
      then follows that $F$ is now split into only three fragments,
      two of which are already $4$-happy. We transform the remaining
      fragment (parallelogram $Z$) into a $2$-happy parallelogram as
      before. We now obtain three fragments, of total complexity
      $O(n)$), even if $F$ contains three or more colinear vertices.
    \end {proof}

  \subsubsection {Computation}

The previous analysis leads to an algorithm, as we show next.
In the following, recall that $z$ is the slope of a line orthogonal to the line through the two start points of the portal edges.

 \begin{lemma}
   \label{lem:compute_parallelogram}
   Let $F$ be a parallelogram with two parallel portal edges, we can
   compute an equivalent $2$-happy parallelogram $F'$ in $O(1)$    time.
 \end{lemma}

  \begin{proof}
    By Lemma~\ref{lem:parallelogram} there exists a parallelogram $F'$
    equivalent to $F$ that has shift zero, and is thus $2$-happy. The
    main task is now to compute such an equivalent fragment.

    %
    To compute the new, equivalent fragment $F'$, we
    could explicitly generate the ray $\rho$, and compute its
    intersection points with the portal edges, tracing the ray until it hits a non-portal edge of the fragment.
    See Figure~\ref{fig-full:parallel}(a). However, this would
    result in a running time linear in the number of such
    intersections. Instead, we can compute the resulting fragment directly, by exploiting the geometry of the situation.
    
    Let $b_l$ be the line segment between the left endpoints of $\self e$ and $\twin e$, and let $b_r$ be defined analogously for the right endpoints of $\self e$ and $\twin e$. See Figure~\ref{fig-full:parallel}(b).
We observe that since the ray $\rho$ is orthogonal to $b_l$, the total length of the ray until it hits a fragment edge is the same as the distance $|\rho|$ between the two parallel lines that go through $b_l$ and $b_r$, respectively. See Figure~\ref{fig-full:parallel}(b).
After cutting along the ray and gluing  the pieces together, the zero-shift guarantees that the resulting shape is a rectangle. The rectangle has $b_l$ as base, and has height equal to the length of the ray until hitting $t$, $|\rho|$. 
This rectangle can be computed in $O(1)$ time by considering the supporting lines of  $\self e$ and $\twin e$, and intersecting them with two orthogonal lines going through the endpoints of $b_l$; refer to the dashed rectangle in Figure~\ref{fig-full:parallel}(b). 
The four intersections of these four lines give the vertices of the rectangle, which is the equivalent 2-happy portalgon.
    %
    %
    %
%
%
%
  \end{proof}

    \begin{figure}[tb]
\centering
      \includegraphics[page=1]{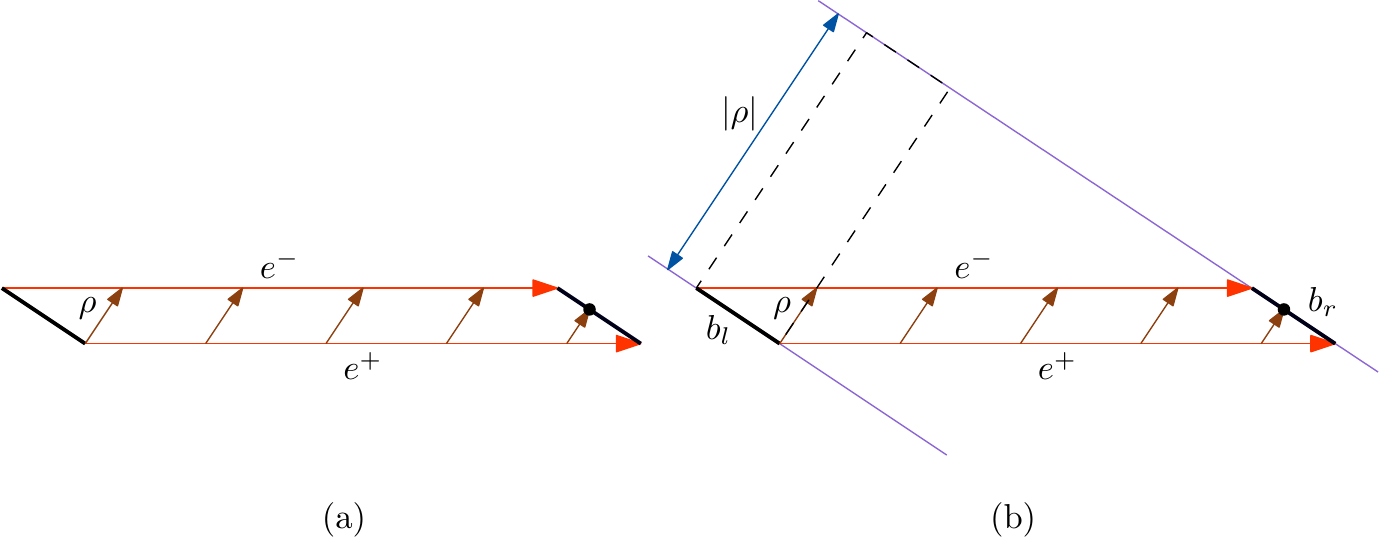}
      \caption{Cutting procedure for two parallel portals. (a) Initial situation. (b) The resulting happy portalgon is the rectangle indicated with dashes.}
      \label{fig-full:parallel}
    \end{figure}

Next we address the case where the fragment is not a parallelogram.

  \begin{lemma}
    \label{lem:compute_leftmost_and_rightmost_segments}
    Let $F$ be a simple fragment containing two horizontal portal
    edges $\self e$ and $\twin e$, and let $z$ be a slope. In $O(n_F)$ time
    we can compute the leftmost segment $s_\ell$ and the rightmost
    segment $s_r$ with slope $z$ with one endpoint on $\self e$ and one
    endpoint on $\twin e$, or report that no such segment exists.
  \end{lemma}

  \begin{proof}
    Assume without loss of generality that the portal edges
    $\self e=\overline{pq}$ and $\twin e=\overline{st}$ are oriented from
    left to right, with $\self e$ above $\twin e$. Let $S$ be the
    parallelogram defined by $\self e$ and $\twin e$, and let $X$ denote
    the part of the boundary of $F$, excluding the interior of $\self e$ and
    $\twin e$ intersecting $S$, i.e.
    $X = S \cap (\partial F \setminus (\self e \cup \twin e \setminus
    \{p,q,s,t\}))$. Since $F$ is simple, $X$ is a collection of
    polygonal chains. There are three types of chains, depending on
    their endpoints: A chain is of type $L$ if its endpoints lie on
    the left boundary $\overline{ps}$ of $S$, of type $R$ if its
    endpoints lie on the right boundary $\overline{qt}$ of $S$, and of
    $M$ if one of its endpoints lies on $\overline{ps}$ and one of its
    endpoint lies on $\overline{qt}$. Clearly, if there are any chains
    of type $M$, the segments $s_\ell$ and $s_r$ do not exist.

    Observe that $s_\ell$, if it exists, contains a vertex of $F$ on a
    chain of type $L$. In particular, let $\ell$ be the vertex on such
    a chain for which the line $m_\ell$ through $\ell$ with slope $z$
    is rightmost (note that there is a well-defined order, since all
    lines with slope $z$ are parallel (and $z \neq 0$)). Segment
    $s_\ell$ contains $\ell$. Symmetrically, $s_r$, if it exists,
    contains the vertex $r$ on a chain of type $R$ whose line $m_r$ is
    leftmost.

    Finally, observe that if $s_\ell$ exists, the line $m_\ell$ must
    be left of the line $m_r$. Conversely, if $m_\ell$ does not lie
    left of $m_r$, $s_\ell$ does not exist.

    We can compute the chains in $X$, and their type, in linear time
    by traversing the boundary of $F$. Computing $\ell$ and $r$ can
    again be done in linear time. Given $\ell$ and $r$, we can then
    decide if $s_\ell$ and $s_r$ exist, and, if they do, we can compute them
    in constant time.
  \end{proof}

  \begin{lemma}
    \label{lem:compute_parallel}
    Let $\P=(\F,P)$ be a portalgon with $n$ vertices and only
    one portal, whose edges are parallel. An equivalent $2$-happy portalgon $\P'$
    with $|\F|+6$ fragments and total complexity $O(n)$ can be computed in $O(n)$ time.
  \end {lemma}

  \begin{proof}
    We implement the constructive proof of
    Lemma~\ref{lem:happy_portalgon_exists_parallel} using
    Lemmas~\ref{lem:compute_leftmost_and_rightmost_segments} and
    \ref{lem:compute_parallelogram}. There is at most one fragment $F$
    with two portal edges $\self e$ and $\twin e$, both of which are
    parallel. We use
    Lemma~\ref{lem:compute_leftmost_and_rightmost_segments} to find
    two parallel segments whose endpoints lie on $\self e$ and
    $\twin e$. If such segments do not exist, $F$ is already 2-happy
    and we are done. Otherwise, they define a parallelogram $Z$ that
    splits $F$ into at most seven pieces, one of which is $Z$
    itself. By
    Lemma~\ref{lem:compute_leftmost_and_rightmost_segments}, computing
    $Z$ takes linear time, as does the actual splitting. We then use
    Lemma~\ref{lem:compute_parallelogram} to transform $Z$ into a
    2-happy parallelogram in $O(1)$ time. All fragments apart
    from $Z$ are already $2$-happy.
  \end{proof}


  \subsection{Single portal, arbitrary edges}
  \label {app:nonpar}
  
    We extend the results from Section~\ref {sec:par} to non-parallel edges.
    Most of the arguments are similar.
  
     \subsubsection {Analysis}

  For two portal edges that are not parallel, the situation does not change too much. First, we analyze which orientations are actually problematic.
  There is a clear relation between the angle between the portal edges, defined as the angle of their corresponding supporting lines, and the number of times that a shortest path can use it.

  \begin {observation}
  \label{obs:angle_bound}
    A fragment with angle $\alpha>0$ between the two portal edges is $(\lfloor \pi/\alpha \rfloor+1)$-happy.
  \end {observation}

  \begin{proof}
  \rodrigo{Please check if the following is convincing}
    In every iteration the angle of the shortest path with respect to one
    of the portals increases by $\alpha$, or in every iteration it decreases by $\alpha$.
    Therefore, if after $\lfloor \pi / \alpha \rfloor$ iterations the shortest path would go through the portal again, it would have covered an angle of more than $\pi$ in the interior of the fragment, which is not possible. Therefore, no more iterations are possible.
    %
    It follows that the shortest path can have at most ($\lfloor \pi / \alpha \rfloor$+1) connected components in the fragment.
  \end{proof}


  
  
  It is worth pointing out that, as soon as the angle $\alpha< \pi/2$, it is possible to construct a situation where a shortest path goes through the portal three times.
  See Figure~\ref{fig-full:60deg} for an example where we make the start point of the non-horizontal portal almost touch the horizontal portal; the orange path is a shortest path.  
    
  \begin{figure}[tb]
      \centering
      \includegraphics{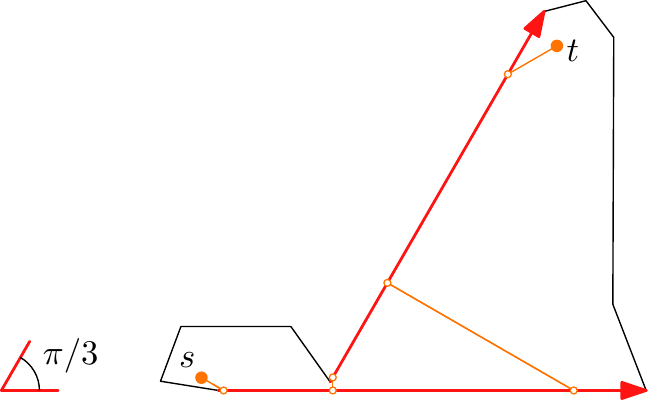}
      \caption{Example for $\alpha= \pi/3 < \pi/2$ where a shortest path (shown orange) goes through the portal three times. 
      }
      \label{fig-full:60deg}
    \end{figure}

  For small $\alpha$, the situation looks quite similar to the one with parallel edges, with some important differences. 
  Let $\Delta$ be the {\em shift} of the portals, defined now as the distance between the projections of the starting vertices of the two portal edges onto the \emph{bisector} of the portal edges, where the bisector is  the angle bisector of  the supporting lines of the edges.

The following result generalizes Lemma~\ref {lem:h0}.

  \begin {lemma}
    \label{cl:nonparallel}
    Let $F$ be a fragment with exactly two portal edges $\self e, \twin e$
    that belong to the same portal. If the shift of $\self e, \twin e$ is
    $0$, then the fragment is $2$-happy.
  \end {lemma}

\begin{proof}

      Let $F$ be a fragment with two portal edges $\self e, \twin e$ with  $\Delta = 0$. Assume without loss of generality that the associated angular bisector is horizontal, and that $\self e$ is above $\twin e$. Refer to Figure~\ref{fig-full:nonpar-zeroshift}.
      Suppose there are two points $p, q \in F$ whose shortest path crosses $e$ twice, say at points $r$ and $s$, in the order $\self r$, $\twin r$, $\self s$, $\twin s$. Consider the part of the shortest path from $\twin r$ to $\self s$. 
Since $\Delta=0$, this path must have length at least $|\twin r\self s|$. 
However, that must be larger than going from $\twin r$ to $\twin s$ on $\twin e$ (i.e., $|\twin  r\self s| > |\twin r \twin s|$), reaching a contradiction.
\end{proof}

  \begin{figure}[h]
\centering
\includegraphics{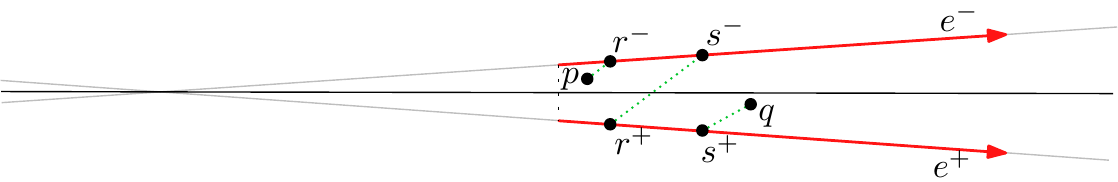}
\caption{Proof of Lemma~\ref{cl:nonparallel}. The zero shift implies that $|\twin r \self s|>|\twin r \twin s|$.}
\label{fig-full:nonpar-zeroshift}
\end{figure}

  \begin {lemma}
  \label{lem:quad_fragment}
    Let $F$ be a quadrilateral where two opposite edges are twin portal edges.
    Then there
      exists a fragment $F'$,  equivalent to $F$, with $\Delta=0$.
    %
    %
\end {lemma}

  \begin {proof}
    We proceed in the same way as in Lemma \ref{lem:parallelogram}, except that now the direction of the ray that we shoot depends also on the angle $\alpha$ between the portal edges.
    Assume without loss of generality
    that $\twin {e}$ is horizontal, $\self e$ is above $\twin e$, and that it starts to the left of
    $\twin e$.  
    Let $z$ be the slope of the line in the direction
    orthogonal to the portal start points, and let $\alpha_z$ be the corresponding angle with the horizontal.
    We begin at the leftmost vertex of $\twin e$, and shoot a ray with angle 
    slope $\alpha_z - \alpha / 2$.
    We consider the ray until it hits the boundary of $F$ for the first time. 
    Every time the ray crosses the portal, we ``cut'' along this ray, creating a new portal along it.
    This results in several smaller fragments,
    which we then glue together again along the pieces of the original portals, into a fragment $F'$.
     See Figure~\ref{fig-full:nonpar4} for an
    illustration. By construction, $F'$ has shift zero.
  \end {proof}
  
\begin{figure}[tb]
      \centering
      \includegraphics{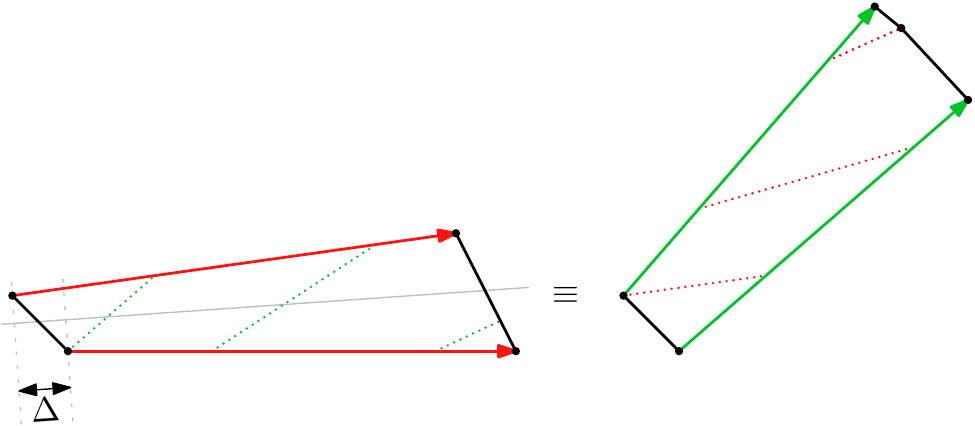}
      \caption{(left) A quadrilateral fragment with non-zero shift
        $\Delta$, and (right) an equivalent fragment with zero shift.}
      \label{fig-full:nonpar4}
    \end{figure}

We need one last lemma, analogous to Lemma \ref{lem:cat}, before we can prove the main result in this section.

    \begin {lemma}
      \label {lem:cat2}
      Consider a fragment $F$ with two portals $\self e$
      and $\twin e$.
         If there is no point $p$ in $e$ such that the interior of the line
      segment $\self p \twin p$  lies in the interior of $F$, then $F$ is
      $5$-happy.
    \end {lemma}
  \begin {proof}
    The proof is almost identical to that of Lemma \ref{lem:cat}.
    
      Suppose that $F$ is not $5$-happy. 
      Assume without loss of generality that $\self e$ is above $\twin e$, and that $\self e$ starts to the left of
    $\twin e$. 
      Then there is a shortest
      path $\geod$ whose intersection with $F$ consists of at least
      six connected components, and thus $\pi$ crosses the portal at
      least five times, in points
      $p_1,\ldots,p_5$. 
      Since $\geod$ is a shortest path, these components cannot share vertices (otherwise we could
      simply shorten the path). 
      It follows that $\geod(\twin {p_2},\self {p_3})$ must be a line segment (otherwise, if it would include a vertex from $F$, such a vertex would be shared with $\geod(\twin {p_1},\self {p_2})$ or with $\geod(\twin {p_3},\self {p_4})$).
      Analogously, $\geod(\twin {p_3},\self {p_4})$ must be a line segment.
      Thus $\twin {p_2} \self {p_3} \self {p_4} \twin {p_3}$
      is an empty quadrilateral. Hence the interior of segment $\self {p_3} \twin {p_3}$ is contained in $F$.
  \end {proof}

  In this case it will be useful to have the following corollary.
  
  \begin{corollary}
  \label{cor:vertical_segment}
      Let \P be a portalgon with one fragment $F$ and exactly one portal, such that the angular bisector between the portal edges is horizontal. 
      If \P is not 5-happy, then there exists a vertical line segment connecting the two portal edges whose interior is inside $F$.
  \end{corollary}
  \begin{proof}
    Assume without loss of generality that 
    the portal edges are diverging, with
    $\self e$ above $\twin e$, and $\self e$ starting to the left of
    $\twin e$. 
   The proof of Lemma~\ref{lem:cat2} guarantees the existence of an empty quadrilateral $\twin {p_2} \self {p_3} \self {p_4} \twin {p_3}$. 
   We have that the four vertices of the empty quadrilateral appear, from left to right, in the order  $\twin {p_2}, \self {p_3}, \twin {p_3}, \self {p_4}$.
   Moreover, the horizontal angular bisector, combined with the fact that  $\self e$ starts to the left of
    $\twin e$, imply that $\self {p_3}$ is strictly to the left of $\twin {p_3}$. 
   Therefore the interior of the vertical line segment that goes upward from $\twin{p_3}$ is in the  interior of the empty quadrilateral, and therefore it is also inside $F$.
  \end{proof}

  \begin {lemma}
  \label{lem:happy_quadri}
     Let \P be a portalgon with one fragment $F$ with $n$
      vertices, and one portal. There exists
      a $5$-happy portalgon $\P'$ equivalent to \P consisting of at
      most three fragments and total complexity $O(n)$.
  \end {lemma}

  \begin{proof}
Assume without loss of generality
    that $\twin {e}$ is horizontal, $\self e$ is above $\twin e$, and that it starts to the left of
    $\twin e$. As before, we consider both portal edges oriented left-to-right.
    
    We now argue
      that when no three vertices of $F$ are collinear, and $F$ is not
      already $5$-happy, we can split $F$ into at most seven $4$-happy
      fragments of total complexity $O(n)$.
      Finally, we show how to reduce the
      number of fragments to three, while remaining $4$-happy, even
      without the general position assumption.

      By Lemma~\ref{lem:cat2}, if there is no line segment of the form $\self{p} \twin{p}$ fully contained in  $F$,  $F$ is already $5$-happy. Let $m_\ell$ (resp., $m_r$) be the
      leftmost (resp., rightmost) such segment; $m_\ell$ contains a vertex $\ell$ of $F$ and $m_r$
      contains a vertex $r$ of $F$ (possibly, $\ell$ or $m$ is an
      endpoint of $\self e$ or $\twin e$).  Let $\self a$ and $\twin a$ be the
      intersection points of $m_\ell$ with $\self e$ and $\twin e$, and let
      $\self b$ and $\twin b$ be the intersection points of $m_r$ with $\self e$
      and $\twin e$.  We cut the quadrilateral
      $Z = \self a \self b \twin a \twin b$ from $F$, which splits $F$ into at most
      seven fragments (since, by general position,
      $\overline{\self a\twin a}$ and $\overline{\self b\twin b}$ contain at
      most two reflex vertices each). We now transform $Z$ into a
      $2$-happy fragment using Lemma~\ref{lem:quad_fragment}.
      %
      Let $T$ and $B$ be the fragments containing the starting points
      of $\self e$ and $\twin e$, respectively. We argue that $T$ is
      $4$-happy. The argument that $B$ is $4$-happy is symmetric. The
      same holds for the fragments containing the endpoints of $\self e$ and
      $\twin e$. Any other fragments (if they exist) contain only one
      portal edge and are thus $2$-happy by
      Observation~\ref{obs:single_portaledge_fragments}.

      Consider the maximal connected components of a shortest path
      $\geod = \geod(s,t)$ with $F$. By
      Observation~\ref{obs:shortest_path_component} such a component
      either: (i) contains $s$, (ii) contains $t$, or (iii) connects a
      point $p_i$ on $\self e$ to a point $q_i$ on $\twin e$. Again by
      Observation~\ref{obs:shortest_path_component} each such
      component can intersect $\overline{\self a\twin a}$ at most once, so
      each such component can intersect $T$ at most once.

      We now further classify the type (iii) components into three
      types, depending on whether $p_i$ lies on the part of $\self e$ in $T$
      and whether $q_i$ lies on the part of $\twin e$ in $B$.

      If $p_i$ lies outside $T$, then
      Observation~\ref{obs:shortest_path_component} implies that
      $\geod(p_i,q_i)$ does not intersect (the interior of) $T$ at all
      since $\geod(p_i,q_i)$ would have to intersect
      $\overline{a\ell}$ twice.

      If $p_i$ lies in $T$ and $q_i$ lies in $B$, then
      Observation~\ref{obs:shortest_path_component} implies that
      $\geod(p_i,q_i)$ contains point $\ell$. Hence, there can be at
      most one such component of this type that intersects $T$.

      If $p_i$ lies in $T$ but $q_i$ lies outside of $B$, we have that
      $q_i = \twin t$ for some point $\self t$ on $e$. 
      If the component
      starting in $\self t$ is again of type (iii) (i.e., $\self t=p_{i+1}$) it
      cannot intersect $B$ as this would imply that
      $\geod(p_{i+1},q_{i+1})$ intersects $\geod(p_i,q_i)$. Nor can it
      intersect $T$ (since $q_{i+1}$ lies on $\twin e$ it would have
      to intersect $\overline{a\ell}$ twice). The same argument holds
      for any component on $\geod(p_j,q_j)$ with $j > i$. Hence,
      $\geod(p_i,q_i)$ intersects $T$ only if this is the last
      component of type (iii). Clearly, there is only one such a
      component.

      It follows that there are only four components of
      $\geod \cap F$ that intersect $T$, and each such component
      intersects $T$ in only one consecutive subpath. Hence $T$ is
      $4$-happy.

      We conclude that the portalgon \P' that we obtain is
      $4$-happy, equivalent to \P, and has at most seven
      fragments. Furthermore, every vertex of \P appears in at most
      $O(1)$ fragments of $\P'$, and thus $\P'$ has complexity $O(n)$.

      Finally, observe that (before splitting $F$) we can actually
      shift the left and right sides of $Z$ inwards by some
      arbitrarily small $\eps$ (in particular, something smaller than
      the smallest distance between two non-adjacent edges of $Z$). It
      then follows that $F$ is now split into only three fragments,
      two of which are already $4$-happy. We transform the remaining
      fragment (parallelogram $Z$) into a $2$-happy parallelogram as
      before. We now obtain three fragments, of total complexity
      $O(n)$), even if $F$ contains three or more collinear vertices.
\end{proof}

\subsubsection{Computation}

In this section we focus on the case where \P consists of just one fragment with one portal.

\begin {lemma}
  \label{lem:linear_time_rayshooting}
  
  Let $\P=(\F,P)$ be an $h$-happy portalgon with $n$ vertices, one fragment, and one portal.
    An equivalent $5$-happy portalgon $\P'$
    with $3$
    fragments and total complexity $O(n)$ can be computed in $O(n + h)$ time.
  \end {lemma}

  \begin {proof}
    We follow the approach in Lemma \ref{lem:happy_quadri}, splitting the fragment into three parts; left, middle,
    and right so that: (i) left and right have complexity $O(n)$, but
    they are already 4-happy, and (ii) middle, with 
    $O(1)$ complexity and allows for $h$ portal crossings (since \P is $h$-happy). It then
    follows that the cost of tracing the ray through the left and
    right parts (by repeatedly computing the first intersection point)
    is $O(n)$, and the cost of tracing the ray through the middle part
    is $O(h)$. 
    
    What remains to explain is how to compute this partition.

    We apply a constant number of rigid transformations to the fragment so that the angular bisector between the two
    portal edges is horizontal, the portal edges are diverging, and
    the top portal edge starts to the left of the bottom one. 
    Note that the required transformations preserve shortest paths. 

    It follows that the segment connecting the left endpoints has
    negative slope, 
    and thus the ray that we shoot has 
    positive
    slope. Let $p_i$ be the $i^\mathrm{th}$ point where the ray ``enters''
    the fragment 
    (i.e., at the bottom edge),
    and let $q_i$ be the $i^\mathrm{th}$ point where it
    ``exits'' the fragment (i.e., at the 
    top edge).
    
    We observe that $p_{i+1}$ lies right of $q_i$.
    Initially, this is due to the ray having 
    positive slope. Since the
    starting point of the top portal lies 
    to the left
    of the starting point
    of the bottom portal, each entry point is to the right of the immediately previous exit point.
    After
    every portal crossing, the slope of the ray 
    decreases,
    approaching zero.
    As a consequence, the ray only crosses any vertical line once.

    By Corollary~\ref{cor:vertical_segment}, we know that if the fragment is not already 5-happy, there is a vertical line segment between the two portal edges that is fully contained in it.
    Thus we can assume that is the case.
    Then we can cut the fragment along two vertical segments; one through the rightmost point ``protruding
    into'' the quadrilateral defined by the two portal edges from the left, and one through the leftmost point protruding from the right.
    
    We can compute these points by computing the hourglass defined by
    the two portal edges (i.e., computing the shortest paths connecting
    the leftmost and rightmost portal endpoints). This takes $O(n)$
    time.
  \end {proof}

  We can improve the previous running time by avoiding tracing the ray in the middle fragment explicitly.

  \begin {lemma}
    \label{lem:make_one_fragment_happy}
   Let $\P=(\F,P)$ be an $h$-happy portalgon with $n$ vertices, one fragment, and one portal.
    An equivalent $5$-happy portalgon $\P'$
    with $3$ fragments and total complexity $O(n)$ can be computed in $O(n + \log h)$ time.
  \end {lemma}

  \begin {proof}
    Following Lemma~\ref{lem:linear_time_rayshooting}, we split the
    portalgon into three parts in $O(n)$ time.
    The middle piece is the only one that may not be happy yet, therefore we have to apply to it the procedure in Lemma \ref{lem:quad_fragment}.
    This comes down to tracing
    the ray through the middle piece.
    To this end, one can apply Lemma 7.1
    of Erickson and
    Nayyeri~\cite{erickson13tracin_compr_curves_trian_surfac}, presented in the context of annular ray shooting,  
    to trace
    the ray through the middle piece in $O(1 + \log k)$ time, for $k$ the number of intersections between the ray and the portal.
    For completeness, we describe here an more direct procedure that yields the same result.
    Consider a shortest path from $p$ to $q$ that goes through a portal $e$ several times.
    Given a crossing point of the shortest path with portal edge $\self e$, the point on $\twin e$ where the ray crosses $\twin e$ next can be computed by a composition of a translation and a rotation.
    Therefore it can be expressed as a matrix multiplication.
    To compute the last point where the shortest path crosses $e$, we need to figure out the power $k$ that corresponds to the last time the ray fully crosses the fragment.
    We do know know $k$, but we can find it in $O(\log k)$ by doing an exponential search. 
    We also need to know the total length of the ray, to actually construct the new fragment, but this we can get by applying $M^k$ to the starting point, and simply measuring the distance directly (we do not need to know the lengths of individual pieces).
    The results now follows by observing that since \P is $h$-happy, we have $k=O(h)$. 
    \qedhere
    %
    
    
    
    
  \end {proof}

  An interesting open question is whether the dependency on $h$ is necessary at all.

\subsection {Generalization based on fragment graph analysis}
\label {app:FGA}

Recall that we defined the \emph{fragment graph}
$G$ of a portalgon $\P$ as the graph that has a node for every fragment, and two fragments are
connected by a link if there is a portal whose portal edges occur in
these two fragments; see Figure~\ref {fig:fragmentgraph}.\footnote {We use the terms {\em node} and {\em link} for the fragment graph, to distinguish from the vertices and edges of the fragments themselves.}

We now argue how to reduce a portalgon to a simpler portalgon for the purpose of transforming it into a happy one; concretely, we show that if the fragment graph of the original portalgon has at most one cycle, then it is sufficient to consider the case of reconfiguring a single fragment with a single portal.

Since, by Observation~\ref{obs:single_portaledge_fragments} any
fragment with only one portal edge is 2-happy, the simplest
non-trivial case occurs when each fragment in the portalgon has two portal edges.
Indeed, we may otherwise simply remove the leaves.


%


%

\begin {lemma} \label {lem:cutleaves}
  Let $\P$ be a portalgon and $G$ its fragment graph. Let $G'$ be the graph we obtain by removing all leaves from $G$. The corresponding portalgon $\P'$ is happy if and only if $\P$ is happy.
\end {lemma}

\begin {proof}
  Let $F$ be a fragment that is a leaf in $G$; $F$ is connected to the remainder of $\P$ by some portal $e$. By Observation~\ref{obs:single_portaledge_fragments}, $F$ is always happy. Furthermore, since no shortest path crosses $e$ twice, any other fragment $X$ of $\P$ is happy if an only if in $\P$ with $F$ removed and $e$ replaced by a boundary edge, $X$ is happy.
\end {proof}

We now argue that we can always safely put leaves back after rearranging the remaining portalgon.

\begin {lemma} \label {lem:putleaveback}
  Let $\P$ be an arbitrary portalgon with complexity $n$ and happiness $h$, and let $F$ be an unrelated fragment with a single portal edge $e$, of complexity $m$.
  Let $S$ be a collection of boundary edges of $\P$ of total length $|e|$.
  Then the portalgon $\P'$ we obtain by turning the edges of $S$ into portal edges and linking them to pieces of $e$ has complexity $O(n+m)$ and happiness $h$.
\end {lemma}
\begin {proof}
  \maarten {Figure?}
  To argue about the complexity, note that the number of edges in $S$ is at most $O(n)$, and we cut the portal edge $e$ into $|S|$ pieces, so clearly the resulting complexity is $O(n+m)$.
  
  To argue about the happiness, note that any shortest path crosses $e$ at most twice. This is still the case after cutting $e$ into smaller portals. Hence, we can still apply Observation~\ref{obs:single_portaledge_fragments}.
\end {proof}

If $G$ has no more leaves, it must contain at least one cycle. We now argue that if $G$ is a cycle, then we can rearrange it into a happy portalgon.

\begin {lemma}
\label{lem:cycle}
  Let $\P$ be a portalgon and $G$ its fragment graph. If $G$ is a cycle, then either $\P$ is 5-happy or there exists a portalgon $\P'$ consisting of a single fragment such that $\P'$ is happy if and only if $\P$ is happy.
\end {lemma}

\begin {proof}
  Suppose we have a sequence of fragments $F_1, F_2, \ldots, F_k$ such that $F_i$ is connected to $F_{i+1}$ by a portal $e_i$ with $\self e_i$ in $F_i$ and $\twin e_i$ in $F_{i+1}$; finally $F_k$ is connected back to $F_1$ by a portal $e'$.
  We may glue the fragments along the $e_i$ portals, yielding a single (but not necessarily simple) fragment $F'$ with only $e'$ remaining. 
  
  If $F'$ is simple, we are done.
  If $F'$ is not simple, we consider two cases:
  either $F'$ can be made simple by cutting off ``pockets'' (which no shortest path will ever visit more than once), or this cannot be done because some shortest path has self-intersections in the universal cover $\hat \Sigma$.
  \maarten {A figure with the two cases would be nice.}
  In the first case, we can make $F'$ simple, in the second case, we argue that $F'$ is already happy.
  We now formalize the two cases; for this we consider the {\em hourglass} $H$, which consists of the union of all shortest paths between points $\self p$ and $\twin p$ through the interior of $F'$ for all points $p \in e'$.
  \begin {enumerate}
    \item Suppose $H$ is simple. We argue that $F' \setminus H$ is a collection of simple polygons (``pockets'') which are connected to $H$ by a single link; therefore, we can cut them off and turn them into separate fragments which are leaves in the fragment graph. By Lemma~\ref {lem:cutleaves} we can ignore them and we are left with a single simple fragment with a single portal, as desired.
    \item Suppose $H$ is not simple. Then there is no point $p \in e'$ such that $\self p$ sees $\twin p$; hence, by Lemma~\ref {lem:cat2}, $F'$ is already 5-happy.
  \end {enumerate}
  The result follows.\qedhere
\end {proof}
%
%
%
We are now ready to prove the main result of this section.

 \begin{theorem}
   \label {thm:makemehappy}
   Let $\P$ be an $h$-happy portalgon with $n$ vertices and fragment graph $G$, such that $G$ has at most one simple cycle.
   We can transform $\P$ into an equivalent 5-happy
   portalgon $\P'$ of total complexity $O(n)$ in $O(n + \log h)$ time.
 \end{theorem}

  \begin {proof}
    We will transform $\P$ through a sequence of portalgons $\P_i$ and denote the corresponding fragment graph of $\P_i$ by $G_i$.
    \begin {itemize}
      \item Let $\P_1 = \P$ be the original portalgon and $G_1 = G$ its fragment graph.
    
      \item First, we calculate a new graph $G_2$ which has no leaves by iteratively applying Lemma~\ref{lem:cutleaves}.
    The resulting portalgon $\P_2$ has the same happiness as $\P$.
    Since $G$ contains at most one cycle, either $G_2$ is a cycle or it is the empty graph.
    If it is empty, by Lemma~\ref{lem:cutleaves}, $\P$ was already happy; in this case we abort and return $\P$.
     
      \item If $G_2$ consists of one cycle, then by Lemma~\ref {lem:cycle} either $\P_2$ is already happy or we can create a portalgon $\P_3$ consisting of single fragment, such that $\P_3$ has the same happiness as $\P$. In the first case, by Lemma~\ref {lem:cutleaves} $\P$ was also already happy, so again we abort and return $\P$.
      
      \item In the second case, we have a portalgon $\P_3$ with a single fragment and a single portal. We apply Lemma~\ref{lem:make_one_fragment_happy} to find an equivalent 5-happy portalgon $\P_4$ of total complexity $O(n)$  in  $O(n + \log h)$ time.
    
      \item Finally, create a portalgon $\P_5$ by iteratively adding back any leaves of $G$ we removed at the start. By Lemma~\ref{lem:putleaveback}, $\P_5$ still has complexity $O(n)$ and the happiness of $\P_5$ is the same as for $\P_4$, that is, at most $5$. 
    \end {itemize}
    We return $\P' = \P_5$.
  \end {proof}


\rodrigo[found this text]{Second idea: perhaps we can do even more? Can we just ignore all fragments which are already happy, even if they have more portals, and then look at the remainder?}

\maarten {I guess we should be able to handle slightly more than that; i.e., two cycles connected by a single path?}



\bibliography{portalgons}

\end{document}